\documentclass[9pt]{article}

\usepackage{graphicx}
\usepackage[utf8]{inputenc}
\usepackage{amsmath}
\usepackage{amssymb}
\usepackage{amsthm}
\usepackage[top=50pt,bottom=50pt,left=62pt,right=69pt]{geometry}
\newtheorem{theorem}{Theorem}[section]
\newtheorem{lemma}[theorem]{Lemma}
\newtheorem{definition}[theorem]{Definition}

\newtheorem{corollary}[theorem]{Corollary}

\newcommand{\Nat}{{\mathbb N}}

\usepackage{amsmath}
\usepackage{amssymb}
\usepackage{thmtools}
\usepackage[all]{xy}
\usepackage{relsize}
\usepackage{stmaryrd}
\usepackage{lscape}
\usepackage{float}
\usepackage{multicol}

\newcommand{\moto}{\to_{\fimp}}
\newtheorem{examples}{Example}

\newcommand{\undots}[1]{#1_{\!\!\!\!\!\dots}\,\!}
\newcommand{\unline}[1]{#1\!\!\!\!\underline{\,\,\,\,}}

\newcommand{\biggap}{\quad \quad \quad}
\newcommand{\Context}{\Gamma}

\newcommand{\ContextEmpty}{\emptyset}
\newcommand{\ContextCat}[2]{#1 , #2}

\newcommand{\TypeVar}{\alpha}
\newcommand{\ValType}{\mathsf{A}}
\newcommand{\ValTypeT}{\mathsf{B}}
\newcommand{\ComType}{\mathsf{\unline{C}}}
\newcommand{\ComTypeT}{\mathsf{\unline{D}}}
\newcommand{\GenType}{\undots{\mathsf{E}}}
\newcommand{\GenTypeT}{\undots{\mathsf{F}}}
\newcommand{\ThunkType}[1]{\mathsf{U} \, #1}
\newcommand{\UnitType}{\mathsf{1}}
\newcommand{\NatType}{\mathsf{N}}
\newcommand{\SumType}[3]{\mathsf{\Sigma}_{#1 \in #2} #3}
\newcommand{\PairType}[2]{#1 \times #2}
\newcommand{\ForceType}[1]{\mathsf{F} #1}
\newcommand{\FunctType}[2]{#1 \to #2}
\newcommand{\ProdType}[3]{\Pi_{#1 \in #2} #3}

\newcommand{\Vjudge}[3]{#1 \vdash^v #2 : #3}
\newcommand{\Cjudge}[3]{#1 \vdash^c #2 : #3}
\newcommand{\derivation}[2]{\frac{#1}{#2}}
\newcommand{\termvar}{x}	\newcommand{\termvart}{y}

\newcommand{\ValTerm}{V}	\newcommand{\ValTermT}{W}
\newcommand{\ValTermD}{L}
\newcommand{\ComTerm}{\underline{M}} 	\newcommand{\ComTermT}{\underline{N}}
\newcommand{\ComTermD}{\underline{K}}
\newcommand{\GenTerm}{\undots{P}} 	\newcommand{\GenTermT}{\undots{R}} 
\newcommand{\GenTermD}{\undots{Q}}
\newcommand{\unit}{*}
\newcommand{\zero}{\mathsf{Z}}
\newcommand{\succe}[1]{\mathsf{S}(#1)}
\newcommand{\caseN}[4]{\mathsf{case} \,\, #1 \,\, \mathsf{in} \,\, \{#2, \succe{#3} \Rightarrow #4\}}
\newcommand{\valseq}[3]{\mathsf{let} \,\, #1 \,\, \mathsf{be} \,\, #2 \,.\, #3}
\newcommand{\return}[1]{\mathsf{return}(#1)}
\newcommand{\comseq}[3]{#1 \,\, \mathsf{to} \,\, #2 \,.\, #3}
\newcommand{\lambt}[2]{\lambda #1 \,.\, #2}
\newcommand{\appt}[2]{#1 \cdot #2}
\newcommand{\thunk}[1]{\mathsf{thunk}(#1)}
\newcommand{\force}[1]{\mathsf{force}(#1)}
\newcommand{\pair}[2]{(#1,#2)}
\newcommand{\patmat}[2]{\mathsf{pm} \,\, #1 \,\, \mathsf{as} \,\, #2}
\newcommand{\prodlam}[2]{\langle #1 \mid #2 \rangle}
\newcommand{\fix}[1]{\mathsf{fix}(#1)}

\newcommand{\reduct}[2]{#1 \; \rightsquigarrow \; #2}
\newcommand{\Stack}{S}	\newcommand{\Stackt}{Z}
\newcommand{\StackEmpty}{\varepsilon}
\newcommand{\StackReduct}[2]{#1 \; \rightarrowtail \; #2}

\newcommand{\Terms}[1]{\mathit{Terms}(#1)}
\newcommand{\EfOp}{\mathsf{op}}
\newcommand{\Trees}[1]{\mathit{T}_{\Sigma}(#1)}
\newcommand{\Treei}[2]{\mathit{Trees}_{#1}(#2)}
\newcommand{\numeral}[1]{\overline{#1}}

\newcommand{\Predicate}{\Theta}

\newcommand{\formulas}[1]{\mathit{Form}(#1)}
\newcommand{\Quanti}{\mathbb{A}}

\newcommand{\Bool}{\mathbb{B}}
\newcommand{\True}{\textbf{\textit{T}}}
\newcommand{\False}{\textit{\textbf{F}}}

\newcommand{\denote}[1]{\llbracket #1 \rrbracket}

\newcommand{\Qrelator}[1]{\QObser(#1)}
\newcommand{\arbsim}{\mathcal{R}}

\newcommand{\Qmodels}[2]{(#1 \models #2)}
\newcommand{\PrePos}{\sqsubseteq^+}
\newcommand{\PreGen}{\sqsubseteq}
\newcommand{\fimp}{\trianglelefteq}
\newcommand{\fpmi}{\trianglerighteq}

\newcommand{\RightSet}[2]{#1 \uparrow\!\! (#2)}
\newcommand{\prebas}{\preceq}
\newcommand{\Bord}{\preccurlyeq}
\newcommand{\modbel}{\curlyeqprec}
\newcommand{\QBS}{\textsf{QBS}}

\newcommand{\EquForm}[1]{\{#1\}}
\newcommand{\ThunkForm}[1]{\langle #1 \rangle}
\newcommand{\PairForm}[2]{(#1,#2)}
\newcommand{\FunctForm}[2]{(#1 \mapsto #2)}
\newcommand{\ModalForm}[2]{#1 (#2)}
\newcommand{\DisForm}[1]{\bigvee #1}
\newcommand{\ConForm}[1]{\bigwedge #1}
\newcommand{\NegForm}[1]{\neg (#1)}

\newcommand{\StepForm}[2]{#1_{\fpmi #2}}
\newcommand{\StabForm}[1]{\kappa_{#1}}
\newcommand{\FirForm}[1]{\pi_1 #1}
\newcommand{\SecForm}[1]{\pi_2 #1}

\newcommand{\SetVar}{D}

\newcommand{\PosLog}{\mathcal{V}^{+}}
\newcommand{\GenLog}{\mathcal{V}}

\newcommand{\LogPre}[1]{\sqsubseteq_{#1}}
\newcommand{\QObser}{\mathcal{Q}}
\newcommand{\qmodal}{q}

\newcommand{\ComExt}[1]{\widehat{#1}}
\newcommand{\HowClo}[1]{(#1)^{\bullet}}
\newcommand{\OpeExt}[1]{(#1)^{\circ}}

\newcommand{\cphi}{\unline{\phi}}
\newcommand{\cpsi}{\unline{\psi}}
\newcommand{\gphi}{\undots{\phi}}
\newcommand{\gpsi}{\undots{\psi}}

\newcommand{\Dord}{\,\Trees{\fimp}\,}
\newcommand{\Tord}{\leq}
\newcommand{\EffCost}[2]{\mathsf{cost}_{#1}(#2)}
\newcommand{\EfCost}[1]{\mathsf{cost}_{#1}}

\makeatletter
\newcommand{\enumeratext}[1]{
}
\makeatother

  \title{Quantitative Logics for Equivalence of Effectful Programs} \author{Niels Voorneveld\footnote{Email:
		Niels.Voorneveld@fmf.uni-lj.si}
	\footnote{This material is based upon work supported by the Air Force Office of Scientific Research under award number FA9550-17-1-0326.
		This project has received funding from the European Union’s Horizon 2020 research and innovation
		programme under the Marie Skłodowska-Curie grant agreement No 731143.}\\
	University of Ljubljana\\
	Ljubljana, Slovenia
}

\begin{document}
\maketitle

\begin{abstract} 
	In order to reason about effects, we can define quantitative formulas to describe behavioural aspects of effectful programs. These formulas can for example express probabilities that (or sets of correct starting states for which) a program satisfies a property. 
	Fundamental to this approach is the notion of quantitative modality, which is used to lift a property on values to a property on computations.
	Taking all formulas together, we say that two terms are equivalent if they satisfy all formulas to the same quantitative degree.
	Under sufficient conditions on the quantitative modalities, this equivalence is equal to a notion of Abramsky's applicative bisimilarity, and is moreover a congruence.
	We investigate these results in the context of Levy's call-by-push-value with general recursion and algebraic effects.
	In particular, the results apply to (combinations of) nondeterministic choice, probabilistic choice, and global store.
\end{abstract}
%

\section{Introduction}\label{section:introduction}
There are many notions of program equivalence for languages with effects. 
In this paper, we explore the notion of \emph{behavioural equivalence}, which states that programs may be considered \emph{behaviourally equivalent} if they satisfy the same behavioural properties. 
This can be made rigorous by defining a logic, where each formula $\phi$ denotes a certain behavioural property. 
We write $\Qmodels{\GenTerm}{\phi}$ to express the satisfaction of formula $\phi$ by term $\GenTerm$, which is usually given by a Boolean truth value (true or false).
Two terms $\GenTerm$ and $\GenTermT$ are said to be behaviourally equivalent if they satisfy the same formulas.
Such an approach is taken in for example \cite{Henessy85}.

In particular, we use this method to define equivalence for a language with \emph{algebraic effects} in the sense of Plotkin and Power \cite{effect}. Effects can be seen as aspects of computation which involves interaction with the world `outside' the environment in which the program runs. They include: exceptions, nondeterminism, probabilistic choice, global store, input/output, cost, etc. The examples given have common ground in the work of Moggi \cite{monad}, and can moreover be expressed by specific effect triggering operations making them `algebraic' in nature.
In the presence of such algebraic effects, computation terms need not simply reduce to a single \emph{terminal term} (that is a \emph{value}), they may also invoke effects on the way. Following \cite{effect,op_meta}, we consider a computation term to evaluate to an \emph{effect tree}, whose nodes are effect operators and leaves are terminal terms. The paper \cite{modal} introduced modalities that lift boolean properties of values to boolean properties of the trees modelling their computations. 
See \cite{Pitts91,Moggi94,Pretnar08} for alternative ways in which logics can be used to describe properties of effects.

The use of a Boolean logic does however not readily adapt to several examples of effects, for example the combination of probability and nondeterminism.
The literature on compositional program verification shows the usefulness of quantitative (e.g. real-number valued) program logics for verifying programs with probabilistic behaviour, possibly in combination with nondeterminism \cite{Kozen1985,McIver2004}.
The paper \cite{modal} develops a general Boolean-valued framework which, although featuring many examples, does not apply to the combination of probability and nondeterminism.

This paper provides a general framework for \emph{quantitative} logics for expressing behavioural properties of programs with effects, generalising the Boolean-valued framework from \cite{modal}.
We consider a quantitative (quantity-valued) satisfaction relation `$\models$', where $\Qmodels{\GenTerm}{\phi}$ is given by an element from a quantitative truth space $\Quanti$ (a \emph{degree} of satisfaction). 
This allows us to ask open questions about programs, like  ``\emph{What is the probability that ...}" or  ``\emph{What are the correct global starting states for ...}".
We define equivalence by stating that programs $\GenTerm$ and $\GenTermT$ are equivalent, if for any formula $\phi$ we have $\Qmodels{\GenTerm}{\phi} = \Qmodels{\GenTermT}{\phi}$ ($\GenTerm$ satisfies $\phi$ precisely as much as $\GenTermT$ does).
A key feature of the logic is the use of \emph{quantitative modalities} to lift quantitative properties on value types to quantitative properties on computation types.

As in \cite{modal}, we are able to establish that the behavioural equivalence defined as above is a \emph{congruence}, as long as suitable properties on the quantitative modalities are satisfied. These properties require notions of \emph{monotonicity}, \emph{continuity}, and a notion of preservation over sequencing called \emph{decomposability}. As in \cite{modal}, the congruence is established by proving that given one of the properties (leaf-monotonicity), our behavioural equivalence is equal to an effect-sensitive 
notion of Abramsky's \emph{applicative bisimilarity} \cite{Abramsky90,Relational}. Given further properties on the modalities, this relation can be proven to be compatible using Howe's method \cite{How}.

The main contribution of this paper is the generalisation of \cite{modal}, and the corresponding generalised results.
This goes through smoothly, though there are some subtleties like what to take as primitive in a quantitative setting. In particular, we will see the necessity of a threshold operation. 
The other main contributions are the examples illustrating the quantitative approach.
Some examples such as the combination of nondeterminism with probabilistic choice, or with global store, do not fit into the Boolean-valued framework of \cite{modal}, but do work here\footnote{The combination of global store and nondeterminism is possible in the framework of \cite{modal}, of one only considers \emph{angelic} (helpful) nondeterminism. The problem is with general (neutral) or demonic (antagonistic) nondeterminism, combined with global store.}$\!$. But there are also examples, such as probability, global store, and cost, whose treatment is more natural in our quantitative setting, even though they also fit in the framework of \cite{modal}. 

As a vehicle of our investigation we use Levy's call-by-push-value (CBPV) \cite{CBPV,LevyThesis}, together with general recursion and the aforementioned algebraic effects. As such, it generalises \cite{modal} in a second way by using call-by-push-value to incorporate both call-by-name (CBN) and call-by-value (CBV) evaluation strategies. This is significant, since once either divergence or effects are present, the distinction between the reduction strategies becomes vital.
For example, if we take some probabilistic choice $\textsf{por}$ signifying a fair coin flip, we have that `$\textsf{por}(\lambt{x}{\numeral{0}}, \lambt{x}{\numeral{1}}) \equiv \lambt{x}{\textsf{por}(\numeral{0}, \numeral{1})}$' holds  in CBN, but not in CBV.
So it is interesting to consider CBPV, as it expresses both these behaviours.
The distinction is expressed in the difference between \emph{production}-types $\ForceType{\ValType}$ 
where one explicitly observes effects, and types like $\FunctType{\ValType}{\ComType}$ 
where the observation of effects is postponed to a later moment.
As such, this language is an ideal backdrop for studying effects.

In Section \ref{section:operational} we give the operational semantics of the language, starting with the effect-free version and working towards our treatment of algebraic effects. In Section \ref{section:logic} we present our quantitative logic, introducing quantitative modalities to deal with the observation of effects. In Section \ref{section:equivalence} we look at the resulting behavioural equivalence and the properties that establish the congruence property (or \emph{compatibility} in its technical form). In Section \ref{section:bisimilarity} we relate this equivalence to applicative (bi)similarities by defining a \emph{relator} using our modalities. This then allows us to adapt a Howe's method proof of compatibility from \cite{Relational,modal} for this equivalence. We finish in Section \ref{section:conclusion} with some discussions.
\section{Operational semantics}\label{section:operational}
We use a simply-typed call-by-push-value functional language as in \cite{LevyThesis,CBPV}, together with general recursion and a ground type for natural numbers, making it a call-by-push-value variant of PCF \cite{PCF}. To this, we add algebraic-effect-triggering operators in the sense of Plotkin and Power \cite{effect}. We first focus on the effect-free part of the language, as we want to consider effects independently of the underlying language.

\subsection{The language}
We give a brief overview of the language and its semantics. The types are divided into two flavours, \emph{Value types} and \emph{Computation types}. Value types contain value terms that are \emph{passive}, they don't compute anything on their own. Computation types contain computation terms which are \emph{active}, which means they either return something to or ask something of the environment.

Value types $\ValType, \ValTypeT$ and computation types $\ComType, \ComTypeT$ are given by:
$$\ValType, \ValTypeT ::= \ThunkType{\ComType} \mid \UnitType \mid \NatType \mid \SumType{i}{I}{\ValType_i} \mid \PairType{\ValType}{\ValType} \quad \quad \quad \quad \quad \quad \ComType, \ComTypeT ::= \ForceType{\ValType} \mid \FunctType{\ValType}{\ComType} \mid \ProdType{i}{I}{\ComType_i}$$

\noindent
where $I$ is any \emph{finite} indexing set. By asserting finiteness of $I$ in the case of product types, the number of program terms is kept countable (a property which will have benefits later on in the formulation of the logic).
The type $\ThunkType{\ComType}$ is a thunk type, which consists of terms which are \emph{frozen}. These terms were initially computation terms but are made inactive by packaging them into a \emph{thunk}.
The type $\NatType$ is the type of natural numbers, containing the non-negative integers. With this type, we can program any computable function on the natural numbers as in PCF \cite{PCF}.
The type $\ForceType{\ValType}$ is a \emph{producer} type, which actively evaluates and \emph{returns} values of type $\ValType$ to the current environment. As was stated, this is the type at which we can observe effects. 
The type $\FunctType{\ValType}{\ComType}$ is a type of functions, which is a computation type since its terms are actively awaiting input.

We have a countably-infinite collection of term variables $x$, and term \emph{contexts}: \quad
$\Context ::= \ContextEmpty \mid \ContextCat{\Context}{x : \ValType}$.

\noindent
Note that contexts only contain Value types, meaning that like in call-by-value, we can only ever substitute value terms. This is no loss of generality, as we can simulate substituting computation terms by packaging them into a thunk.
The terms of the language are as follows:
\begin{align*}
	\textbf{Value} & \textbf{ terms: } \ValTerm, \ValTermT ::= \unit \mid \zero \mid \succe{\ValTerm} \mid x \mid \thunk{\ComTerm} \mid \pair{i}{\ValTerm} \mid \pair{\ValTerm}{\ValTermT}\\
	\textbf{Computation} & \textbf{ terms: } \ComTerm, \ComTermT ::= \caseN{\ValTerm}{\ComTerm}{x}{\ComTermT} \mid \valseq{x}{\ValTerm}{\ComTerm} \mid \return{\ValTerm} \mid \comseq{x}{\ComTerm}{\ComTermT} \mid \force{\ValTerm} \\
	    \mid \lambt{x:\ValType}{\ComTerm} & \mid  \appt{\ComTerm}{\ValTerm} \mid \patmat{\ValTerm}{\{\dots,(i.x)\,.\,\ComTerm_i,\dots\}}
	 \mid \patmat{\ValTerm}{(x,y)\,.\,\ComTerm} \mid \prodlam{\ComTerm_i}{i \in I} \mid \appt{\ComTerm}{i} \mid \fix{\ComTerm}
\end{align*}
\noindent
We underline terms ($\ComTerm$) and types ($\ComType$) when they are computation terms and computation types respectively. We will also use $\GenType, \GenTypeT$ and $\GenTerm, \GenTermT$ to denote general types and their terms, e.g. they could be either value or computation types/terms.
Following \cite{LevyThesis}, their typing rules are given in Fig.\,\ref{fig:typing}. We distinguish two typing judgements, $\vdash^v$ and $\vdash^c$, for value and computation terms respectively. We write $\Terms{\GenType}$ for the set of closed terms of type $\GenType$. Note the addition of the fixpoint operator $\fix{-}$ which has been added to allow for general recursion and hence divergence. We write $\numeral{n} : \NatType$ for the numeral representing the $n$-th natural number.

\begin{figure}{\small
	$$
	\derivation{ }
	{\Vjudge{\Context}{\unit}{\UnitType}}
	\biggap
	\derivation{ }
	{\Vjudge{\Context}{\zero}{\NatType}}
	\biggap
	\derivation{\Vjudge{\Context}{\ValTerm}{\NatType}}
	{\Vjudge{\Context}{\succe{\ValTerm}}{\NatType}}
	\biggap
	\derivation{\Vjudge{\Context}{\ValTerm}{\NatType} \biggap \Cjudge{\Context}{\ComTerm}{\ComType} \biggap 
		\Cjudge{\ContextCat{\Context}{\termvar : \NatType}}{\ComTermT}{\ComType}}
	{\Cjudge{\Context}{\caseN{\ValTerm}{\ComTerm}{\termvar}{\ComTermT}}{\ComType}}$$
	$$\derivation{ }
	{\Vjudge{\ContextCat{\ContextCat{\Context}{\termvar : \ValType}}{\Context'}}{\termvar}{\ValType}} 
	\quad
	\derivation{\Vjudge{\Context}{\ValTerm}{\ValType} \biggap \Cjudge{\ContextCat{\Context}{\termvar : \ValType}}{\ComTerm}{\ComType}}
	{\Cjudge{\Context}{\valseq{\termvar}{\ValTerm}{\ComTerm}}{\ComType}}
	\quad
	\derivation{\Vjudge{\Context}{\ValTerm}{\ValType}}
	{\Vjudge{\Context}{\return{\ValTerm}}{\ForceType{\ValType}}} 
	\quad
	\derivation{\Cjudge{\Context}{\ComTerm}{\ForceType{\ValType}} \biggap \Cjudge{\ContextCat{\Context}{\termvar : \ValType}}{\ComTermT}{\ComType}}
	{\Cjudge{\Context}{\comseq{\ComTerm}{\termvar}{\ComTermT}}{\ComType}}$$
	$$\derivation{\Cjudge{\Context}{\ComTerm}{\ComType}}
	{\Vjudge{\Context}{\thunk{\ComTerm}}{\ThunkType{\ComType}}} 
	\biggap
	\derivation{\Vjudge{\Context}{\ValTerm}{\ThunkType{\ComType}}}
	{\Cjudge{\Context}{\force{\ValTerm}}{\ComType}}
	\biggap
	\derivation{\Cjudge{\ContextCat{\Context}{\termvar : \ValType}}{\ComTerm}{\ComType}}
	{\Cjudge{\Context}{\lambt{\termvar:\ValType}{\ComTerm}}{\FunctType{\ValType}{\ComType}}}
	\biggap
	\derivation{\Vjudge{\Context}{\ValTerm}{\ValType} \biggap \Cjudge{\Context}{\ComTerm}{\FunctType{\ValType}{\ComType}}}
	{\Cjudge{\Context}{\appt{\ComTerm}{\ValTerm}}{\ComType}}
	$$
	$$
	\derivation{\Vjudge{\Context}{\ValTerm}{\ValType_j}}
	{\Vjudge{\Context}{\pair{j}{\ValTerm}}{\SumType{i}{I}{\ValType_i}}}j \in I
	\quad
	\derivation{\Vjudge{\Context}{\ValTerm}{\SumType{i}{I}{\ValType_i}} \biggap \Cjudge{\ContextCat{\Context}{\termvar : \ValType_i}}{\ComTerm_i}{\ComType} \ \ \text{for each} \ \ i \in I}
	{\Cjudge{\Context}{\patmat{\ValTerm}{\{\dots,(i.\termvar)\,.\,\ComTerm_i,\dots\}}}{\ComType}}
	\quad
	\derivation{\Vjudge{\Context}{\ValTerm}{\ValType} \biggap \Vjudge{\Context}{\ValTerm'}{\ValTypeT}}
	{\Vjudge{\Context}{\pair{\ValTerm}{\ValTerm'}}{\PairType{\ValType}{\ValTypeT}}}
	$$
	$$
	\derivation{\Vjudge{\Context}{\ValTerm}{\PairType{\ValType}{\ValTypeT}} \biggap \Cjudge{\ContextCat{\ContextCat{\Context}{\termvar : \ValType}}{\termvart : \ValTypeT}}{\ComTerm}{\ComType}}
	{\Cjudge{\Context}{\patmat{\ValTerm}{\pair{\termvar}{\termvart}\,.\,\ComTerm}}{\ComType}}
	\quad\,\,
	\derivation{\Cjudge{\Context}{\ComTerm_i}{\ComType_i} \ \ \text{for each} \ \ i \in I}
	{\Cjudge{\Context}{\prodlam{\ComTerm_i}{i \in I}}{\ProdType{i}{I}{\ComType_i}}}
	\quad\,\,
	\derivation{\Cjudge{\Context}{\ComTerm}{\ProdType{i}{I}{\ComType_i}}}
	{\Cjudge{\Context}{\appt{\ComTerm}{j}}{\ComType_j}}
	\quad\,\,
	\derivation{\Cjudge{\Context}{\ComTerm}{\FunctType{\ThunkType{\ComType}}{\ComType}}}
	{\Cjudge{\Context}{\fix{\ComTerm}}{\ComType}}$$
}
	\caption{Typing rules}
	\label{fig:typing}
\end{figure}

\subsection{Semantics}
We give the semantics of this language by specifying a reduction strategy for computation terms in the style of a CK-machine \cite{CK}. 
We distinguish a special class of computation terms, called \emph{terminal terms}, which will not reduce further. They consist of: $\return{\ValTerm} : \ForceType{\ValType}$, $\lambt{\termvar:\ValType}{\ComTerm} : \FunctType{\ValType}{\ComType}$, and $\prodlam{\ComTerm_i}{i \in I} : \ProdType{i}{I}{\ComTerm_i}$.

We first give the rules for terms we can directly reduce. 
We denote these using relation symbol $\reduct{}{}$:

\begin{multicols}{2}
\begin{enumerate}
	\item $\reduct{\caseN{\zero}{\ComTerm}{\termvar}{\ComTermT}}{\ComTerm}$.
	\item $\reduct{\caseN{\succe{\ValTerm}}{\ComTerm}{\termvar}{\ComTermT}}{\ComTermT[\ValTerm/\termvar]}$.
	\item $\reduct{\valseq{\termvar}{\ValTerm}{\ComTerm}}{\ComTerm[\ValTerm/\termvar]}$.
	\item $\reduct{\force{\thunk{\ComTerm}}}{\ComTerm}$.
	\item $\reduct{\patmat{\pair{j}{\ValTerm}}{\{\dots,(i.\termvar)\,.\,\ComTerm_i,\dots\}}}{\ComTerm_j[\ValTerm/\termvar]}$.
	\item $\reduct{\patmat{\pair{\ValTerm}{\ValTerm'}}{\pair{\termvar}{\termvart}\,.\,\ComTerm}}{\ComTerm[\ValTerm/\termvar,\ValTerm'/\termvart]}$.
	\item $\reduct{\fix{\ComTerm}}{\appt{\ComTerm}{\thunk{\fix{\ComTerm}}}}$.
\end{enumerate}
\end{multicols}

The behaviour of the other non-terminal computation terms; $\comseq{\ComTerm}{\termvar}{\ComTermT}$, $\appt{\ComTerm}{\ValTerm}$ and $\appt{\ComTerm}{i}$, is implemented using a system of \emph{stacks} defined recursively: \quad $\Stack, \Stackt ::= \StackEmpty \mid \Stack \circ \comseq{(-)}{\termvar}{\ComTerm} \mid \Stack \circ \ValTerm \mid \Stack \circ j$.

\noindent
We write $\Stack\{\ComTerm\}$ for the computation resulting from applying $\Stack$ to $\ComTerm$, which can be seen as evaluating the program $\ComTerm$ within the environment $\Stack$.

$\StackEmpty\{\ComTerm\} := \ComTerm$ \biggap \biggap \biggap $(\Stack \circ \comseq{(-)}{\termvar}{\ComTermT})\{\ComTerm\} := \Stack\{\comseq{\ComTerm}{\termvar}{\ComTermT}\}$

$(\Stack \circ \ValTerm)\{\ComTerm\} := \Stack\{\appt{\ComTerm}{\ValTerm}\}$ \biggap
$(\Stack \circ i)\{\ComTerm\} := \Stack\{\appt{\ComTerm}{i}\}$

\noindent
Whenever one encounters a computation of which one needs to first evaluate a subterm, one unfolds the continuation into the Stack and focusses on evaluating that subterm.
This method is given by the stack reduction relation $\StackReduct{}{}$
in the following way:
\begin{multicols}{2}
\begin{enumerate}
	\item If $\reduct{\ComTerm}{\ComTermT}$, then $\StackReduct{(\Stack,\ComTerm)}{(\Stack,\ComTermT)}$.
	\item $\StackReduct{(\Stack,\comseq{\ComTerm}{\termvar}{\ComTermT})}
	{(\Stack \circ \comseq{(-)}{\termvar}{\ComTermT}, \ComTerm)}$.
	\item $\StackReduct{(\Stack \circ \comseq{(-)}{\termvar}{\ComTermT}, \return{\ValTerm})}
	{(\Stack,\ComTermT[\ValTerm/\termvar])}$.
	\item $\StackReduct{(\Stack,\appt{\ComTerm}{\ValTerm})}
	{(\Stack \circ \ValTerm, \ComTerm)}$.
	\item $\StackReduct{(\Stack \circ \ValTerm, \lambt{\termvar:\ValType}{\ComTerm})}
	{(\Stack, \ComTerm[\ValTerm / \termvar])}$.
	\item $\StackReduct{(\Stack,\appt{\ComTerm}{j})}
	{(\Stack \circ j, \ComTerm)}$.
	\item $\StackReduct{(\Stack \circ j, \prodlam{\ComTerm_i}{i \in I})}
	{(\Stack, \ComTerm_j)}$.
\end{enumerate}
\end{multicols}

\subsection{Adding algebraic effect operators}

We add algebraic effects in the style of \cite{op_meta}, given by specific effect operators.
We use a type variable $\TypeVar$ for computation types.  
Effects are given by operators of the following arities (like in \cite{op_meta,Plotkin:2009}):

$\TypeVar^n \to \TypeVar \mid \NatType \times \TypeVar^n \to \TypeVar \mid \TypeVar^{\NatType} \to \TypeVar$.

\noindent
For each effect under consideration, we bundle together effect operators pertaining that effect in a set called an \emph{effect signature} $\Sigma$. Given such a signature, new computation terms can be constructed according to the typing rules in Fig. \ref{fig:effect_typing}.

\begin{figure}[H]{\small
	$$
	\derivation{\forall 1 \leq i \leq n.(\Cjudge{\Context}
		{\ComTerm_i}{\ComType})}
	{\Cjudge{\Context}{\EfOp (\ComTerm_1,\dots,\ComTerm_n)}{\ComType}} ~ \EfOp : \TypeVar^n \to \TypeVar
	\biggap \biggap
	\derivation{\Cjudge{\Context, x : \NatType}
		{\ComTerm}{\ComType}}
	{\Cjudge{\Context}{\EfOp (x \,.\, \ComTerm)}{\ComType}} ~ \EfOp : \TypeVar^{\NatType} \to \TypeVar
	$$
	$$
	\derivation{\Vjudge{\Context}{\ValTerm}{\NatType} \biggap \forall 1 \leq i \leq n.(\Cjudge{\Context}
		{\ComTerm_i}{\ComType})}
	{\Cjudge{\Context}{\EfOp (\ValTerm, \ComTerm_1,\dots,\ComTerm_n)}{\ComType}} ~ \EfOp : \NatType \times \TypeVar^n \to \TypeVar
	$$}
	\vspace{-5mm}
	\caption{Effect typing rules}
	\label{fig:effect_typing}
\end{figure}

We look at the running examples of this paper.
\begin{examples}[Probabilistic choice]\label{example:prob}
	The effect of probability is implemented by the signature $\Sigma_{p} := {\{\mathsf{por} : \alpha^2 \to \alpha\}}$, where we have a single binary operator for fair probabilistic choice. The computation $\mathsf{por}(\ComTerm,\ComTermT)$ will have a $\frac{1}{2}$ probability of evaluating $\ComTerm$, and $\frac{1}{2}$ probability of evaluating $\ComTermT$.
\end{examples}
\begin{examples}[Global store]\label{example:GS}
	In the case of global variables for natural numbers, we define a signature $\Sigma_{g} := \bigcup_{l \in L} {\{\textsf{lookup}_l : \alpha^{\NatType} \rightarrow \alpha, \textsf{update}_l : \NatType \times \alpha \to \alpha\}}$, where we have a set of locations $L$ for storing natural numbers. The computation $\textsf{lookup}_l(x \,.\, M)$ looks up the number at location $l$ and substitutes it for $x$ in $M$. The computation $\textsf{update}_l(\overline{n}, M)$ stores $n$ at $l$ and then continues with the computation $M$. 
\end{examples}
\begin{examples}[Probabilistic choice and global store]\label{example:prob+GS}
	We will also consider the combination of the previous two examples, probabilistic choice with global store, given by effect signature $\Sigma_{pg} := \Sigma_p \cup \Sigma_g$.
\end{examples}
\begin{examples}[Cost]
	If we want to keep track of costs of an evaluation, we take the signature $\Sigma_c := {\{\EfCost{c} : \alpha \to \alpha \mid c \in C \}}$, where we have a countable set of real-valued costs $C$. The computation $\EffCost{c}{M}$ assigns a cost of $c$ to the evaluation of $M$. This cost can represent a time delay or some other resource.
\end{examples}
\begin{examples}[Combinations with nondeterminism]\label{example:probdet}
	We consider a binary operator $\textsf{nor}: \alpha^2 \rightarrow \alpha$ for nondeterministic choice, which contrary to probabilistic choice is entirely unpredictable. One interpretation is to consider it under the control of some external agent or scheduler (e.g. a compiler), which one may wish to model as being cooperative (\emph{angelic}), antagonistic (\emph{demonic}), or \emph{neutral}.
	We will consider nondeterminism and it's operator in combination with any one of the previous three examples. The resulting signatures are named $\Sigma_{pn}$, $\Sigma_{gn}$, $\Sigma_{gpn}$, and $\Sigma_{cn}$ respectively.
\end{examples}
\begin{examples}[Combinations with error]\label{example:error}
	Lastly, given some set of error messages $E$, we consider adding error raising effect operations $\{\mathsf{raise}_e: \alpha^0 \to \alpha \mid e \in E\}$ to the language, where $\mathsf{raise}_e()$ stops the evaluation of a term, and displays message $e$. There is no continuation possible afterwards.
\end{examples}

In the presence of such effects, the evaluation of a computation term might halt when encountering an algebraic effect operator. We broaden the semantics, where a computation term now evaluates to an \emph{effect tree}, a coinductively generated term using operations from our effect signature $\Sigma$ together with terminal terms and a symbol for divergence $\bot$. This idea appears in \cite{effect}, but here we adapt the formulation from \cite{op_meta} to call-by-push-value. 

We define the notion of an \emph{effect tree} over any set $X$, where $X$ can be thought of as a set of terminal terms.
\begin{definition}\label{def:tree}
	An \emph{effect tree} (henceforth \emph{tree}) over a set $X$, determined by a signature $\Sigma$ of effect operations,  is a labelled and possibly infinite depth tree whose nodes have the possible forms given below.
	\begin{enumerate}
		\item A leaf node labelled $\bot$ (the symbol for nontermination/divergence).
		\item A leaf node labelled $x$ where $x \in X$.
		\item A node labelled `$\EfOp$' with children $t_1,\dots, t_n$, when $\EfOp\in \Sigma$ has arity $\alpha^n \rightarrow \alpha$.
		\item A node labelled `$\EfOp$' with children $t_0,t_1,\dots$, when $\EfOp\in \Sigma$ has arity $\alpha^{\NatType} \rightarrow \alpha$.		
		\item A node labelled `$\EfOp_m$' where $m \in \Nat$ with children $t_1,\dots, t_n$, when $\EfOp\in \Sigma$ has arity $\NatType \times \alpha^n \rightarrow \alpha$.
	\end{enumerate}
\end{definition}

\noindent
The set of trees over set $X$ and signature $\Sigma$ is denoted $\Trees{X}$. We can equip this set with a partial order $\Tord$, where $t \Tord r$ if $r$ can be constructed from $t$ by pruning (possibly infinitely many) subtrees and labelling the pruning points with $\bot$.
Moreover, the preorder is \emph{$\omega$-complete}, so each ascending chain of trees $t_0 \Tord t_1 \Tord \dots$ has a least upper bound $\sqcup_n t_n$.

For any $x \in X$, we denote $\eta(x) \in \Trees{X}$ for the tree which only consists of one leaf labelled $x$. We also have a map $\mu : \Trees{\Trees{X}} \to \Trees{X}$ which flattens a tree of trees into one tree, by transforming the leaves (which are trees) into subtrees.

For each computation type $\ComType$ we define the evaluation map $|-| : \Terms{\ComType} \to \Trees{\Terms{\ComType}}$, which returns a tree, whose leaves are either labelled with $\bot$ or labelled with a terminal term of type $\ComType$. We define this inductively by constructing for each $n \in \Nat$ the $n$-th approximation of the tree.
\begin{enumerate}
	\item $|\Stack, \ComTerm|_0 := \bot$
	\item $|\StackEmpty, \ComTerm|_{n+1} := \eta(\ComTerm)$ if $\ComTerm$ is a terminal computation term.
	\item $|\Stack, \ComTerm|_{n+1} := |\Stack', \ComTerm'|_n$ if $\StackReduct{(\Stack, \ComTerm)}{(\Stack', \ComTerm')}$.
	\item $|\Stack, \EfOp(\ComTerm_1,\dots,\ComTerm_m)|_{n+1} := \EfOp\{|\Stack, \ComTerm_1|_n, \dots, |\Stack, \ComTerm_m|_n\}$ if $\EfOp : \TypeVar^m \to \TypeVar$.
	\item $|\Stack, \EfOp(x \,.\, \ComTerm)|_{n+1} := \EfOp\{m \mapsto |\Stack, \ComTerm[\numeral{m}/x]|\}$  if $\EfOp : \TypeVar^{\NatType} \to \TypeVar$.
	\item $|\Stack, \EfOp(\numeral{k}, \ComTerm_1,\dots,\ComTerm_m)|_{n+1} := \EfOp_k\{|\Stack, \ComTerm_1|_n, \dots, |\Stack, \ComTerm_m|_n\}$ if $\EfOp : \NatType \times \TypeVar^m \to \TypeVar$.
\end{enumerate}
Using this, we define $|M| := \bigsqcup_n |\StackEmpty, M|_n$. We view $|M|$ as an operational semantics of $M$ in which $M$ is reduced to its (possibly) observable computational behaviours, namely the tree of effect operations potentially performed in the evaluation of $M$. See Figure \ref{fig:trees} for two examples of effect trees.

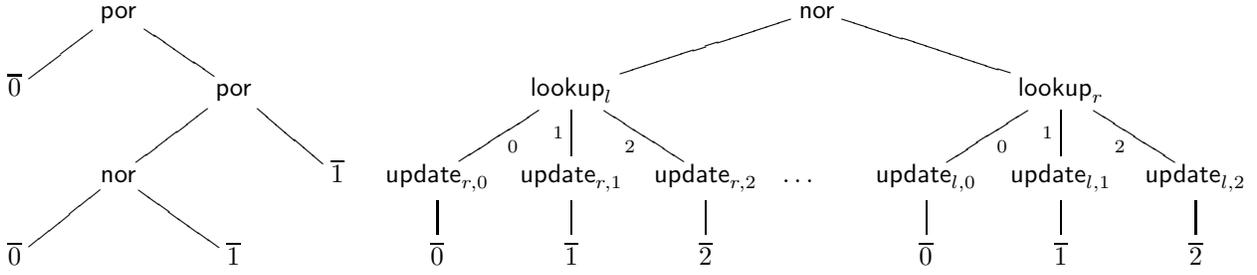
\begin{figure}
	$\xymatrix@R=1.5em{
		& \textsf{por} \ar@{-}[ld] \ar@{-}[rd] & & 
		& & & & \textsf{nor} \ar@{-}[lld] \ar@{-}[rrd] & & &
		\\
		\overline{0} & & \textsf{por} \ar@{-}[ld] \ar@{-}[rd] & 
		& & \textsf{lookup}_l \ar@{-}[ld]^0 \ar@{-}[d]_1 \ar@{-}[rd]_2 & & & & \textsf{lookup}_r \ar@{-}[ld]^0 \ar@{-}[d]_1 \ar@{-}[rd]_2 &
		\\
		& \textsf{nor} \ar@{-}[ld] \ar@{-}[rd] & & \overline{1} 
		& \hspace{-5mm}\textsf{update}_{r,0}\hspace{-5mm} \ar@{-}[d] & \hspace{-5mm}\textsf{update}_{r,1}\hspace{-5mm} \ar@{-}[d] & \hspace{-5mm}\textsf{update}_{r,2}\hspace{-5mm} \ar@{-}[d] & \hspace{-5mm}\dots & \hspace{-5mm}\textsf{update}_{l,0}\hspace{-5mm} \ar@{-}[d] & \hspace{-5mm}\textsf{update}_{l,1}\hspace{-5mm} \ar@{-}[d] & \hspace{-5mm}\textsf{update}_{l,2}\hspace{-5mm} \ar@{-}[d]
		\\
		\overline{0} & & \overline{1} &  
		& \overline{0} & \overline{1} & \overline{2} & & \overline{0} & \overline{1} & \overline{2}
		\\
	}$
	\caption{Tree examples: The unpredictable coin and the nondeterministic copier.}
	\label{fig:trees}
\end{figure}

These trees are still quite syntactic, and may contain lots of unobservable information irrelevant to the real-world behaviour of programs.
In the next section, we will set up the quantitative logic which will extract from such trees only the relevant information, using quantitative modalities.

\section{Quantitative Logic}\label{section:logic}
We define a quantitative logic expressing \emph{behavioural properties} of terms. Each type has a set of formulas, which can be satisfied by terms of that type to varying \emph{degrees of satisfaction}. These degrees of satisfaction are given by truth values from a complete lattice.

A \emph{countably complete lattice} is a set $\Quanti$ with a partial order $\fimp$, where for each subset $X \subseteq \Quanti$ there is a least upper bound $\sup(X)$ and a greatest lower bound $\inf(X)$. In particular, we define $\True := \sup(\Quanti) = \inf(\emptyset)$ as the completely true value, and $\False := \inf(\Quanti) = \sup(\emptyset)$ as the completely false value.

We also equip this space with a notion of \emph{negation} or involution, which is a bijective map $\neg : \Quanti \to \Quanti$ such that $\forall a \in \Quanti, \neg(\neg a) = a$ and $\forall a, b \in \Quanti, (a \fimp b) \Leftrightarrow (\neg b \fimp \neg a)$.
We will use the words involution and negation interchangeably.
Given the conditions of an involution, it holds that $\neg \True = \False$ and $\neg \False = \True$. \footnote{Results about the positive behavioural equivalence, and the positive logic, do not need negation. So a subset of results in this paper are still valid without it.}
Examples of complete lattices with involution/negation used in this paper are:
\begin{enumerate}
	\item The \emph{Booleans} $\Bool := \{\True, \False\}$, where $\False \fimp \True$. Negation swaps the elements.
	\item The \emph{real number interval} $[0,1]$, with the usual order. Here $\True = 1$, $\False = 0$, and $\neg x := 1-x$.
	\item The powerset $\mathcal{P}(X)$ over some set $X$, whose order is given by inclusion $\subseteq$, so $\True = X$ and $\False = \emptyset$. Negation is given by the complement, where $\neg A := X-A = \{x \in X \mid x \notin A\}$.
	\item For $A$ a complete lattice and $X$ a set,  the function space $A^X$ with point-wise order is a complete lattice.
	\item The \emph{infinite interval} $[0,\infty]$ with reversed order. Here, $\True = 0$, $\False = \infty$, and $\neg x := 1/x$.
\end{enumerate}

We construct a logic for our language in order to define a behavioural preorder. For each type $\GenType$, value or computation, we have a set of formulas $\formulas{\GenType}$. Greek letters $\phi, \psi, \dots$ are used for formulas over value types, underlined Greek letters $\cphi, \cpsi, \dots$ for formulas over computation types, and underdotted Greek letters $\gphi, \gpsi, \dots$ for formulas over any type. We are aiming to define a quantitative relation $\Qmodels{\GenTerm}{\gphi}$ to denote the element of $\Quanti$ which describes the degree to which the term $\GenTerm$ satisfies the formula $\gphi$ (e.g. this may describe the probability of satisfaction or the amount of time needed for satisfaction).
We choose the formulas according to the following two design criteria, as in \cite{modal}.

Firstly, we design our logic to only contain \emph{behaviourally meaningful} formulas. This means we only want to test properties that are in some sense observable by users and/or other programs. For example, for the natural numbers type $\NatType$ we have a formula $\EquForm{n}$ which checks whether a term is equal to the numeral $\numeral{n}$.
For function types we have formulas of the form $\FunctType{\ValTerm}{\cphi}$ which tests a program on a specific input $\ValTerm$, and checks how much the resulting term satisfies $\cphi$.

Secondly, we desire our logic to be as \emph{expressive} as possible. To this end, we add countable disjunction (suprema) $\bigvee$ and conjunction (infima) $\bigwedge$ over formulas, together with negation $\neg$. Moreover, we add two natural quantity-specific primitives: a threshold operation and constants. Both such operations are used frequently (albeit implicitly) in practical examples of quantitative verification, e.g. in \cite{McIver2004}.

\subsection{Quantitative modalities}
Fundamental to the design of the logic is how we interpret algebraic effects. In CBPV, effects are observed in producer types $\ForceType{\ValType}$. In order to formulate \emph{observable} properties of $\ForceType{\ValType}$-terms in our logic, we include a set of \emph{quantitative modalities} which lift formulas on $\ValType$ to formulas on $\ForceType{\ValType}$.
We bundle our a selection of quantitative modalities together in a set $\QObser$. 

Each modality $\qmodal \in \QObser$ denotes a function $\denote{\qmodal} : \Trees{\Quanti} \to \Quanti$, which is used to translate a \emph{tree of truths} into a singular truth value. Given a quantitative predicate $\Predicate : X \to \Quanti$ on a set $X$, we can use a modality $q$ to lift it to a quantitative predicate $\ModalForm{\qmodal}{\Predicate}: \Trees{X} \to \Quanti$ as follows. For $t \in \Trees{X}$, we write $t[\Predicate]$ for the tree in $\Trees{\Quanti}$ where each non-$\bot$ leaf $x \in X$ is replaced by the value $\Predicate(x)$. Then $\ModalForm{\qmodal}{\Predicate}(t) = \denote{q}(t[\Predicate])$.

In the examples, we will define the denotation of a modality $q$ by giving for each $n \in \Nat$ an approximation $\denote{q}_n$.
These will follow the rules: $\denote{q}_0(t) = \False$, $\denote{q}_n(\bot) = \False$, and $\denote{q}_{n+1}(\eta(a)) = a$, and effect specific rules given in the examples below.
Given these approximations, the denotation $\denote{q}(t)$ is given by $\sup\{\denote{q}_n(t) \mid n \in \Nat\}$.
\setcounter{examples}{0}
\begin{examples}[Probabilistic choice]
	We use as quantitative truth space the real number interval $\Quanti := [0,1]$ with $\fimp \,:=\, \leq$, which denote probabilities\footnote{One could alternatively consider $[0,\infty]$ as the value set.}. We take a single modality $\mathsf{E}$ for expectation, where $\Qmodels{\ComTerm}{\ModalForm{\mathsf{E}}{\Predicate}}$ gives the expected value of $\Qmodels{\ValTerm}{\Predicate}$ given the probabilistic distribution $\ComTerm$ induces on its return values $\ValTerm$.
	This is achieved by giving $\mathsf{E}$ the denotation $\denote{\mathsf{E}} : \Treei{\Sigma_p}{\Quanti} \to \Quanti$ which sends a tree of real numbers to the expected real number, where the approximation of the denotation is given by:
	$\denote{\mathsf{E}}_{n+1}(\textsf{p-or}(t,r)) = (\denote{\mathsf{E}}_n(t)+\denote{\mathsf{E}}_n(r))/2$.
\end{examples}
\begin{examples}[Global store]
	Given a set of locations $L$, we have a set of \emph{states} $S := L \to \Nat$. Our set of truth values is given by the powerset $\Quanti := \mathcal{P}(S)$ with $\fimp \,:=\, \subseteq$.
	We have a single modality $\mathsf{G}$, where $\Qmodels{\ComTerm}{\ModalForm{\mathsf{G}}{\Predicate}}$ gives the set of starting states for which $\ComTerm$ terminates with a value $\ValTerm$ such that the end state is contained in $\Qmodels{\ValTerm}{\Predicate}$.
	We define this formally with the following rules:
	$\denote{\mathsf{G}}_{n+1}(\textsf{lookup}_l(t_0, t_1, \dots)) = {\{s \in S \mid s \in \denote{\mathsf{G}}_{\max(0,n-s(l))}(t_{s(l)})\}}$ and
	$\denote{\mathsf{G}}_{n+1}(\textsf{update}_{l,m}(t)) = \{s \mid s[l := m] \in \denote{\mathsf{G}}_{n}(t)\}$.
\end{examples}
\begin{examples}[Probabilistic choice and global store]
	For this combination of effects, we take as truth space the functions $\Quanti := [0,1]^S$ with point-wise order, where $S$ is the set of global states and $[0,1]$ the lattice of probabilities with standard order. Intuitively, this space assigns to each starting state a probability that a property is satisfied. We define a single modality $\mathsf{EG}$ which, for each state $s \in S$, is given by the following rules: $\denote{EG}_{n+1}(\mathsf{p-or}(t,r))(s) := (\denote{EG}_n(t)(s) + \denote{EG}_n(r)(s))/2$, $\denote{\mathsf{EG}}_{n+1}(\textsf{lookup}_l(t_0, t_1, \dots))(s) = \denote{\mathsf{EG}}_{\max(0,n-s(l))}(t_{s(l)})(s)$, and
	$\denote{\mathsf{EG}}_{n+1}(\textsf{update}_{l,m}(t))(s) = \denote{\mathsf{EG}}_{n}(t)(s[l := m])$.
\end{examples}
\begin{examples}[Cost]
	We use the infinite real number interval $\Quanti := [0,\infty]$ with $\fimp \,:=\, \geq$ denoting an abstract notion of \emph{cost} (e.g. time). 
	Trees are just branches in this example.
	We have a single modality $\mathsf{C}$, where $\Qmodels{\ComTerm}{\ModalForm{\mathsf{C}}{\Predicate}}$ is the cost it takes for $\ComTerm$ to evaluate plus the cost given by $\Qmodels{\ValTerm}{\Predicate}$, where $\ValTerm$ is the value $\ComTerm$ evaluates to.
	The definition of $\denote{\mathsf{C}} : \Treei{\Sigma_d}{\Quanti} \to \Quanti$ is such that $\denote{\mathsf{C}}_{n+1}(\EffCost{c}{t}) = c+\denote{\mathsf{C}}_n(t)$. Note that for any tree $t$ either infinite or with leaf $\bot$, we have $\denote{\mathsf{C}}(t) = \infty$.
	This reflects the idea that a diverging computation will exceed any possible finite cost.
\end{examples}
\begin{examples}[Combinations with nondeterminism]
	To add nondeterminism to any of the previous examples, we keep their truth space $\Quanti \in \{[0,1], \mathcal{P}(S), [0,1]^X, [0,\infty]\}$, and extend the definition of their modality $q \in \{\mathsf{E}, \mathsf{G}, \mathsf{EG}, \mathsf{C}\}$ in two ways, creating an optimistic modality $q_{\lozenge}$ and a pessimistic modality $q_{\Box}$ such that:
	
	$\denote{q_{\lozenge}}_{n+1}(\textsf{n-or}(t,r)) = \denote{q_{\lozenge}}_n(t) \vee \denote{q_{\lozenge}}_n(r) \biggap \biggap \denote{q_{\Box}}_{n+1}(\textsf{n-or}(t,r)) = \denote{q_{\Box}}_n(t) \wedge \denote{q_{\Box}}_n(r)$.
	
	For the combination with probability, we can see the nondeterministic choice as being controlled by some external agent, which chooses a \emph{strategy} for resolving the nondeterministic choice nodes, like in a Markov decision processes. 
	$\mathsf{E}_{\lozenge}$ finds the optimal strategy to get the best expectation, whereas $\mathsf{E}_{\Box}$ finds the worst strategy.
	Similarly, $\mathsf{C}_{\lozenge}$ will search for the minimum possible execution cost, while $\mathsf{C}_{\Box}$ will look for the maximum cost.
	
	For instance, if the denotation $|\ComTerm|$ of a term $\ComTerm$ of type $\ForceType{\NatType}$ is given by the first tree in Fig. \ref{fig:trees}, then $\Qmodels{\ComTerm}{\ModalForm{\mathsf{E}_{\lozenge}}{\EquForm{1}}} = 1/2$ and $\Qmodels{\ComTerm}{\ModalForm{\mathsf{E}_{\Box}}{\EquForm{1}}} = 1/4$.
	If $|\ComTermT|$ is given by the second tree from Fig. \ref{fig:trees}, then $\Qmodels{\ComTermT}{\ModalForm{\mathsf{G}_{\lozenge}}{\EquForm{0}}} = \{s \in S \mid s(l) = 0 \vee s(r) = 0\}$ and $\Qmodels{\ComTermT}{\ModalForm{\mathsf{G}_{\Box}}{\EquForm{0}}} = \{s \in S \mid s(l) = 0 = s(r)\}$.
\end{examples}

\begin{examples}[Combinations with error]
	There are two ways of defining combinations with error messages, akin to the sum and tensor approach of combining effects from e.g.~\cite{Hyland:2006}.
	Let $\Sigma$, $\Quanti$ and $\QObser$ be the signature, truth space, and quantitative modalities of the effects to which we want to add error messages from a set $E$.
	Given a modality $q \in \QObser$ and some function $f: E \to \Quanti$, assigning to each message a value, we define a new modality $q_f$ which, besides inheriting the rules from $q$, follows the rule $\denote{q_f}_{n+1}(\mathsf{raise}_e) = f(e)$.
	We define two new sets of modalities for this combination, $\QObser^+ := \{q_f \mid q \in \QObser, f: E \to \{\False, \True\}\}$ and $\QObser^{\times} := \{q_f \mid q \in \QObser, f: E \to \Quanti\}$, each giving a different interpretation of error.
	
	E.g. in the presence of global store (Example \ref{example:GS}), the modalities from $\QObser^+$ are not able to observe the final global state when an error message has been raised, whereas some modalities from $\QObser^{\times}$ can. For instance, for $e \in E$ and $f: E \to \Quanti$ such that $f(e) := \{s[l:=1] \mid s \in S\}$, it holds that $\mathsf{G}_f$ is in $\QObser^{\times}$ but not in $\QObser^+)$. Moreover, $(\textsf{update}_l(\overline{1}, \textsf{raise}_e()) \models (\ModalForm{\mathsf{G}_f}{\top})) = \True$ whereas $(\textsf{update}_l(\overline{0}, \textsf{raise}_e()) \models \ModalForm{\mathsf{G}_f}{\top}) = \False$. Those two terms are however not distinguishable by any modality from $\QObser^+$.
\end{examples}

All the Boolean-valued examples of modalities for effects in \cite{modal}, can also be accommodated in our quantitative setting by taking $\Quanti := \{\True,\False\}$. These include for instance Input/Output.

\subsection{Formulation of the logic}
We write $\formulas{\GenType}$ for the set of formulas over type $\GenType$, which is defined by induction on the structure of $\GenType$.
Fig. \ref{fig:formula} gives the inductive rules for generating these formulas. We have modality formulas $\ModalForm{q}{\phi}$, constant formulas $\StabForm{a}$, and step formulas $\StepForm{\gphi}{a}$. Note that conjunctions and disjunctions (i.e., meets and joins) are taken over countable sets of formulas only.
\begin{figure}
	{
		$$
		\derivation{n \in \Nat}
		{\EquForm{n} \in \formulas{\NatType}}
		\biggap
		\derivation{\cphi \in \formulas{\ComType}}
		{\ThunkForm{\cphi} \in \formulas{\ThunkType{\ComType}}}
		\biggap
		\derivation{\phi \in \formulas{\ValType_j}}
		{\PairForm{j}{\phi} \in \formulas{\SumType{i}{I}{\ValType_i}}}
		\biggap
		\derivation{\phi \in \formulas{\ValType}}
		{\FirForm{\phi} \in \formulas{\PairType{\ValType}{\ValTypeT}}}
		$$
		$$
		\derivation{\phi \in \formulas{\ValTypeT}}
		{\SecForm{\phi} \in \formulas{\PairType{\ValType}{\ValTypeT}}}
		\biggap
		\derivation{\ValTerm \in \Terms{\ValType} \biggap \cphi \in \formulas{\ComType}}
		{\FunctForm{\ValTerm}{\cphi} \in \formulas{\FunctType{\ValType}{\ComType}}}
		\biggap
		\derivation{j \in I \biggap \cphi \in \formulas{\ComType_j}}
		{\FunctForm{j}{\cphi} \in \formulas{\ProdType{i}{I}{\ComType_i}}}
		$$
		$$
		\derivation{\qmodal \in \QObser \biggap \phi \in \formulas{\ValType}}
		{\ModalForm{\qmodal}{\phi} \in \formulas{\ForceType{\ValType}}}
		\biggap
		\derivation{X \subseteq_{\text{countable}} \formulas{\GenType}}
		{\DisForm{X} \in \formulas{\GenType}}
		\biggap
		\derivation{X \subseteq_{\text{countable}} \formulas{\GenType}}
		{\ConForm{X} \in \formulas{\GenType}}
		$$
		$$
		\derivation{\gphi \in \formulas{\GenType} \biggap a \in \Quanti}
		{\StepForm{\gphi}{a} \in \formulas{\GenType}}
		\biggap
		\derivation{a \in \Quanti}
		{\StabForm{a} \in \formulas{\GenType}}
		\biggap
		\derivation{\gphi \in \formulas{\GenType}}
		{\neg \gphi \in \formulas{\GenType}}
		$$
	}
	\caption{Formula constructors}
	\label{fig:formula}
\end{figure}

We define a quantitative satisfaction relation $\models : \Terms{\GenType} \times \formulas{\GenType} \to \Quanti$, thus for $\GenTerm \in \Terms{\GenType}$ and $\phi \in \formulas{\GenType}$, the \emph{satisfaction} statement $\Qmodels{\GenTerm}{\phi}$ denotes an element of $\Quanti$.
Satisfaction of the formulas is defined inductively by the following rules:

\medskip
\begin{tabular}{r l r l}
$\Qmodels{\ValTerm}{\EquForm{n}}$ & $:= \,\,\,\begin{cases} \True & \text{if } \ValTerm = \numeral{n}\\ \False & \text{otherwise} \end{cases} \biggap$ & $\Qmodels{\pair{i}{\ValTerm}}{\PairForm{j}{\phi}}$ & $:= \,\,\,\begin{cases} \Qmodels{\ValTerm}{\phi} & \text{ if } i=j.\\
	\False & \text{ otherwise.}
\end{cases}$\\
$\Qmodels{\ValTerm}{\ThunkForm{\cphi}}$ & $:= \,\,\,\Qmodels{\force{\ValTerm}}{\cphi}$ &
$\Qmodels{\pair{\ValTerm_1}{\ValTerm_2}}{\FirForm{\phi}}$ & $:= \,\,\, \Qmodels{\ValTerm_1}{\phi}$.\\
$\Qmodels{\pair{\ValTerm_1}{\ValTerm_2}}{\SecForm{\phi}}$ & $:= \,\,\, \Qmodels{\ValTerm_2}{\phi}$. &
$\Qmodels{\ComTerm}{\FunctForm{\ValTerm}{\cphi}}$ & $:= \,\,\, \Qmodels{\appt{\ComTerm}{\ValTerm}}{\cphi}$.\\
$\Qmodels{\ComTerm}{\FunctForm{j}{\cphi}}$ & $:= \,\,\, \Qmodels{\appt{\ComTerm}{j}}{\cphi}$. & $\Qmodels{\ComTerm}{\ModalForm{\qmodal}{\phi}} $ & $:=\,\,\,  \denote{\qmodal}(|\ComTerm|[\phi])$
\end{tabular}
\medskip

\noindent
The modality formula $\ModalForm{q}{\phi}$ is particularly important, as it expresses how the quantitative modalities are used to observe effects. 
The last couple of satisfaction rules are for formula constructors occurring at each type. 
\begin{align*}
\Qmodels{\GenTerm}{\DisForm{X}} := \sup\{\Qmodels{\GenTerm}{\gphi} \mid \gphi \in X\}.
\biggap \Qmodels{\GenTerm}{\ConForm{X}} := \inf\{\Qmodels{\GenTerm}{\gphi} \mid \gphi \in X\}.\\
\Qmodels{\GenTerm}{\StepForm{\gphi}{a}} := \begin{cases} \True & \text{if } \Qmodels{\GenTerm}{\gphi} \fpmi a. 
\biggap \biggap \biggap \Qmodels{\GenTerm}{\StabForm{a}} := a.\\ 
\False & \text{otherwise.} 
\biggap \biggap \biggap \Qmodels{\GenTerm}{\NegForm{\gphi}} := \neg \Qmodels{\GenTerm}{\gphi}.
\end{cases}
\end{align*}
All formulas together form the \emph{general logic} $\GenLog$. We distinguish a specific fragment of $\GenLog$, the \emph{positive logic} $\PosLog$ excluding all formulas which use $\NegForm{}$. The logic $\PosLog$ can be interpreted without giving an involution on $\Quanti$.

We end this section by looking at some interesting properties we can construct using the logic, illustrating the expressibility of the logic.
In case of global store (Example \ref{example:GS}), we can construct formulas in the style of Hoare logic.
For instance, taking two subsets $P,Q \in \Quanti = \mathcal{P}(S)$ of global states, the statement $\ComTerm \models \StepForm{\ModalForm{\mathsf{G}}{\StabForm{Q}}}{P}$ will give $\True$, precisely if, when starting the execution of $\ComTerm$ with a state from $P$, the execution will terminate with a state from $Q$. 
As another example, in case of global store with probability (Example \ref{example:prob+GS}), where $\Quanti := [0,1]^S$, we can construct, given a formula $\cphi$ and a distribution of states $\mu \in [0,1]^S$, a formula $\Sigma_{\mu}(\cphi)$ such that $\Qmodels{\ComTerm}{\Sigma_{\mu}(\cphi)}(s) = \min(1,\sum_{s \in S} \mu(s) \cdot\Qmodels{\ComTerm}{\cphi}(s))$. Then $\Qmodels{\ComTerm}{\Sigma_{\mu}(\ModalForm{\mathsf{EG}}{\StabForm{\True}})}$ expresses the probability of termination of $\ComTerm$, given that the starting state is sampled from $\mu$.
In the same vein, we can look at the combination of probability and nondeterminism (Example \ref{example:prob} and \ref{example:probdet}), where $\Qmodels{\ComTerm}{\bigvee_{a,b \in [0,1]}(\StepForm{\ModalForm{\mathsf{E}_{\lozenge}}{\StabForm{\True}}}{a} \wedge \StepForm{\ModalForm{\mathsf{E}_{\Box}}{\StabForm{\True}}}{b} \wedge \StabForm{(a+b)/2})}$ expresses the probability that $\ComTerm$ terminates, given that the agent/scheduler in control of nondeterministic choice is sampled from a distribution of which 50\%  is helpful and 50\% is antagonistic.

\section{Behavioural equivalence}\label{section:equivalence}
We can define a behavioural preorder for any sub-collection of formulas $\mathcal{L}$.
\begin{definition}
	For any fragment of the logic $\mathcal{L} \subseteq \GenLog$, the \emph{logical preorder} $\PreGen_{\mathcal{L}}$ is given by:
	
	$\forall \GenTerm, \GenTermT \in \Terms{\GenType}, \, \GenTerm \PreGen_{\mathcal{L}} \GenTermT \iff (\forall \gphi \in \formulas{\GenType}, \Qmodels{\GenTerm}{\gphi} \fimp \Qmodels{\GenTermT}{\gphi})$
	
	\noindent
	The \emph{logical equivalence} $\equiv_{\mathcal{L}}$ is given by $\PreGen_{\mathcal{L}} \cap \sqsupseteq_{\mathcal{L}}$.
\end{definition}

The \emph{general behavioural preorder} $\PreGen$ is the logical preorder $\PreGen_{\GenLog}$, whereas the \emph{positive behavioural preorder} $\PrePos$ is the logical preorder $\PreGen_{\PosLog}$. We denote $\equiv$ and $\equiv^+$ for the logical equivalences $\equiv_{\GenLog}$ and $\equiv_{\PosLog}$ respectively (the behavioural equivalences). These closed relations can be extended to relations on open terms by using the \emph{open extension} (where two open terms are related if they are related for any substitution of variables).

A \emph{basic} formula is a non-constant formula (not necessarily atomic) which on the top level does not have conjunction $\bigwedge$, disjunction $\bigvee$, negation $\neg$, constant formula $\StabForm{a}$ or step-construction $\StepForm{(-)}{a}$. 
It is not difficult to see that both $\PreGen$ and $\PrePos$ are completely determined by basic formulas.
Note that since $\PosLog \subseteq \GenLog$, it holds that $(\PreGen) \subseteq (\PrePos)$ and $(\equiv) \subseteq (\equiv^+)$.

\begin{lemma}\label{lemma:symmetric}
	The general behavioural preorder $\PreGen$ is symmetric, so $(\PreGen) = (\equiv)$.
\end{lemma}
\begin{proof}
	Assume $\GenTerm \PreGen \GenTermT$, then for any formula $\gphi$ we have $\neg \Qmodels{\GenTerm}{\gphi} = \Qmodels{\GenTerm}{\NegForm{\gphi}} \fimp \Qmodels{\GenTermT}{\NegForm{\gphi}} = \neg \Qmodels{\GenTermT}{\gphi}$. Hence $\Qmodels{\GenTermT}{\gphi} \fimp \Qmodels{\GenTerm}{\gphi}$ for any $\gphi$, so $\GenTermT \PreGen \GenTerm$. 
\end{proof}

\begin{lemma}\label{lemma:char_formula}
	For any fragment $\mathcal{L}$ of the logic $\GenLog$ closed under countable conjunctions, it holds that for any term $\GenTerm : \GenType$ there is a formulas $\chi^{\GenTerm}$ s.t.: If $\GenTerm \PreGen_{\mathcal{L}} \GenTermT$ then $\Qmodels{\GenTermT}{\chi^{\GenTerm}} = \True$, else $\Qmodels{\GenTermT}{\chi^{\GenTerm}} = \False$.
\end{lemma}
\begin{proof}
	For any $\GenTermT : \GenType$ such that $\GenTerm \not\PreGen_{\mathcal{L}} \GenTermT$ we can find a formula $\psi^{\GenTermT}$ such that $\Qmodels{\GenTerm}{\psi^{\GenTermT}} \not\fimp \Qmodels{\GenTermT}{\psi^{\GenTermT}}$. We choose such a formula for each $\GenTermT$ as above, and define $X := \{\StepForm{\psi^{\GenTermT}}{\Qmodels{\GenTerm}{\psi^{\GenTermT}}} \mid \GenTermT : \GenType, \GenTerm \not\PreGen_{\mathcal{L}} \GenTermT\}$, which is countable since there are countably many terms. Then $\chi^{\GenTerm} := \ConForm{X}$ has the desired properties.
\end{proof}

\begin{lemma}\label{lem:pre_char}
	For $\mathcal{L} \in \{\GenLog, \PosLog\}$, we have the following characterisations of the logical preorder $\mathcal{R} :=\, \LogPre{\mathcal{L}}$:
	\begin{multicols}{2}
		\begin{enumerate}
			\item $\ValTerm \mathcal{R}_{\NatType} \ValTermT \iff V=W$.
			\item $\ComTerm \mathcal{R}_{\ThunkType{\ComType}} \ComTermT \iff \force{\ComTerm} \,\mathcal{R}_{\ComType}\, \force{\ComTermT}$.
			\item $\pair{j}{\ValTerm} \mathcal{R}_{\SumType{i}{I}{\ValType_i}} \pair{k}{\ValTermT} \iff (j = k) \wedge \ValTerm \mathcal{R}_{\ValType_j} \ValTermT$.
			\item $\pair{\ValTerm}{\ValTerm'} \mathcal{R}_{\PairType{\ValType}{\ValTypeT}} \pair{\ValTermT}{\ValTermT'} \Leftrightarrow \ValTerm \mathcal{R}_{\ValType} \ValTermT \wedge \ValTerm' \mathcal{R}_{\ValTypeT} \ValTermT'$.
			\item $\ComTerm \mathcal{R}_{\FunctType{\ValType}{\ComType}} \ComTermT \iff \forall V : \ValType, \,\,\appt{\ComTerm}{\ValTerm} \,\mathcal{R}_{\ComType}\, \appt{\ComTermT}{\ValTerm}$.
			\item $\ComTerm \mathcal{R}_{\ProdType{i}{I}{\ComType_i}} \ComTermT \iff \forall j \in I, \,\, \appt{\ComTerm}{j} \,\mathcal{R}_{\ComTerm_j}\, \appt{\ComTermT}{j}$.
		\end{enumerate}
	\end{multicols}
\end{lemma}
\begin{proof}
	Note that at each type level, the preorder is completely determined by basic formulas. All other formulas depend solely on the satisfaction of basic formulas, by a simple induction.
	As such, the above characterisations are a simple consequence of unfolding the satisfaction relation of basic formulas.
\end{proof}

\subsection{Congruence properties}
A relation on terms is \emph{compatible}, if it is preserved over the typing rules from Fig. \ref{fig:typing}.
We introduce the three properties that we will require in order to establish that (the open extensions of) the behavioural preorders are compatible, hence precongruences.
The space $\Trees{\Quanti}$, which forms the basis of the technical definition of the modalities, plays a fundamental role in this. 

The first property considers the leaf order $\Trees{\fimp}$ on $\Trees{\Quanti}$, where $t \,\Trees{\fimp}\, r$ if $r$ can be created by replacing leaves $a \in \Quanti$ of $t$ by leaves $b \in \Quanti$ of higher value $a \fimp b$.
The $\bot$ leaves can however not be replaced.
\begin{definition}
	A modality $q \in \QObser$ is \emph{leaf-monotone} if $\forall t,r \in \Trees{\Quanti}$, $t \,\Trees{\fimp}\, r \implies {\denote{q}(t) \fimp \denote{q}(r)}$.
\end{definition}
\noindent
This property is useful for establishing a variety of different results, but mainly just shows that modalities preserve the implicit (point-wise) order $\gphi \Rightarrow \gpsi$ on formulas ($\gphi \Rightarrow \gpsi$ iff $\forall \GenTerm, \Qmodels{\GenTerm}{\gphi} \fimp \Qmodels{\GenTermT}{\gpsi}$).

The second property considers the $\omega$-complete tree order $\Tord$ on $\Trees{\Quanti}$, defined just after Definition \ref{def:tree}.
\begin{definition}
	A modality $q \in \QObser$ is \emph{tree Scott continuous} if for any ascending chain $t_0 \Tord t_1 \Tord t_2 \Tord \dots$ it holds that $\denote{q}(\bigsqcup_{n \in \Nat} t_n) = \sup\{\denote{q}(t_n) \mid n \in \Nat\}$.
\end{definition}
\noindent
This is property is necessary in the congruence proof for inductively approximating the satisfaction value of infinite trees generated by the fixpoint operator and infinite arity effect operators.

The third and final property is the most technical one, and considers the preservation of the behavioural preorder over sequence operations such as $\comseq{(-)}{x}{(-)}$.
It considers the monad multiplication map $\mu: \Trees{\Trees{\Quanti}} \to \Trees{\Quanti}$, and requires that the abstract generalisation of the behavioural preorder on $\Trees{\Trees{\Quanti}}$ and $\Trees{\Quanti}$ is preserved by the $\mu$-map.
To formulate this, we need first define these abstract relations.

We write $h: \Quanti \moto \Quanti$ to say that $h$ is a \emph{monotone} function from $\Quanti$ to $\Quanti$, so $a \fimp b \Rightarrow h(a) \fimp h(b)$.
For a function $h: X \to Y$ we write $h^*: \Trees{X} \to \Trees{Y}$ for its (functorial) \emph{lifting} defined by $h^*(t) := t[x \to h(x)]$.
For a function $h: X \to \Quanti$ (a \emph{valuation} on $X$) and a modality $q \in \QObser$, we write $t \in q(h)$ for $\denote{q}(h^*(t))$.
\begin{definition}
	For any two trees $t, t' \in T\Quanti$,
	$t \prebas t' \,\, \iff \,\, \forall q \in \QObser, \forall h: \Quanti \moto \Quanti, ~ (t \in q(h)) \fimp (t' \in q(h))$.
\end{definition}

For any relation $\mathcal{R} \subseteq X \times Y$, and valuation $h: X \to \Quanti$, we define $(\RightSet{\mathcal{R}}{h}): Y \to \Quanti$ to be the function such that $\RightSet{\mathcal{R}}{h}(b) := \sup_{a \in X,a \mathcal{R} b}(h(a))$.
We classify abstract quantitative behavioural properties on $T\Quanti$.
A function $H: \Trees{\Quanti} \to \Quanti$ is called \emph{quantitative behaviourally saturated} if for any two trees $t, t' \in T\Quanti$ such that $t \prebas t'$, it holds that $H(t) \fimp H(t')$.
We write $\QBS$ for the set of quantitative behaviourally saturated functions.
Note that $H \in \QBS$ if and only if there is a function $F: T\Quanti \to \Quanti$ such that $H = \RightSet{\prebas}{F}$.
Moreover, for any $q \in \QObser$, it is easy to see that $\denote{q} \in \QBS$.
We define a relation on \emph{quantitative double trees} $TT\Quanti$.
\begin{definition}
	We define the preorder $\modbel$ on $TT\Quanti$ by: \,for any two quantitative double trees $r,r' \in TT\Quanti$,
	$r \modbel\, r' \quad \iff \quad \forall q \in \QObser, \forall H \in \QBS, ~ r \in q(H) ~ \fimp ~ r' \in q(H)$.
\end{definition}
\begin{lemma}\label{lem:quanmodbel_equivalences}
	If all modalities $q \in \QObser$ are leaf-monotone, then for any two $r,r' \in TT\Quanti$,\\ 
	$r \modbel r' \quad \iff \quad \forall F: \Trees{\Quanti} \to \Quanti, \forall q \in \QObser, ~ r \in q(F) ~ \fimp ~ r' \in q(\RightSet{\prebas}{F})$.
\end{lemma}
\begin{proof}
	For `$\Rightarrow$', note that for any $t \in T\Quanti$, $F(t) \fimp \RightSet{\prebas}{F}(t)$ so the result follows from leaf-monotonicity and the fact that $\RightSet{\prebas}{F} \in \QBS$.
	For `$\Leftarrow$', use that for $H \in \QBS$, $\RightSet{\prebas}{H} = H$.
\end{proof}

We can define the third property, decomposability, together with its stronger counterpart, sequentiality\footnote{Sequentiality is one of two properties for $\denote{q}$ to be an Eilenberg-Moore algebra for the monad $\Trees{-}$}.
\begin{definition}
	$\QObser$ is \emph{decomposable} if for all $t,r \in \Trees{\Trees{\Quanti}}$, if $t \modbel r$ then $\mu t \prebas \mu r$. A modality $q \in \QObser$ is \emph{sequential} if for all $t \in \Trees{\Trees{\Quanti}}, \denote{q}(\mu t) = \denote{q}(\denote{q}^*(t))$.
\end{definition}

\begin{lemma}
	If all modalities $q \in \QObser$ are leaf-monotone and sequential, then $\QObser$ is decomposable.
\end{lemma}
\begin{proof}
	Assume $r \modbel r'$, and let $q \in \QObser$ and $h: \Quanti \moto \Quanti$.
	It is easy to see that $\denote{q} \in \QBS$, so $(\denote{q} \circ h^*) \in \QBS$ too. 
	Since $q$ is sequential, $\denote{q}(h^*(\mu r)) = \denote{q}(\mu(h^{**}(r))) = \denote{q}(\denote{q}^* \circ h^{**}(r)) = (r \in q(\denote{q} \circ h^*)) \fimp (r' \in q(\denote{q} \circ h^*))$ which is by the same steps equal to $\denote{q}(h^*(\mu r'))$.
	Hence $\mu r \prebas \mu r'$, and as such, $\QObser$ is decomposable.
\end{proof}

The three properties defined above allow us to establish compatibility:
\begin{theorem}\label{main_theorem}
	If $\QObser$ is a decomposable set of leaf-monotone and Scott tree continuous modalities, then $\PreGen$ and $\PrePos$ are compatible, hence precongruences.
\end{theorem}

All our examples satisfy these three properties. Both leaf-monotonicity and Scott tree continuity are consequences of the inductive and hence continuous definitions of the modalities, while decomposability holds by observing that any modality from the examples is sequential.
We illustrate this in the following lemma.

\begin{lemma}
	In Example \ref{example:probdet} of combined probability and nondeterminism, $\mathsf{E}_{\Box}$ is sequential.
\end{lemma}
\begin{proof}
Take $r,r' \in \Trees{\Trees{\Quanti}}$ as above and assume $\denote{\mathsf{E}_{\Box}}(\mu r) > a \in [0,1]$, then since $\denote{\mathsf{E}_{\Box}}(\mu r) = \sup_n(\denote{\mathsf{E}}_n(\mu r))$ there must be an $n \in \mathbb{N}$ such that $\denote{\mathsf{E}}_n(\mu r) > a$. By the recursive definition of $\denote{\mathsf{E}_{\Box}}_{(-)}$ we can see that $\denote{\mathsf{E}}_n(r[t \mapsto F'_n(t)]) \geq \denote{\mathsf{E}}_n(\mu r)$, and hence $\denote{\mathsf{E}_{\Box}}(\denote{\mathsf{E}_{\Box}}^*(r)) > a$. 

Now assume $\denote{\mathsf{E}_{\Box}}(\denote{\mathsf{E}_{\Box}}^*(r)) > a$, then there must be an $m$ such that $\denote{\mathsf{E}}_m(\denote{\mathsf{E}_{\Box}}^*(r)) > a$. Now, $\denote{\mathsf{E}}_m$ only looks at a finite amount of leaves, and hence there must be a $k$ such that $\denote{\mathsf{E}}_m(\denote{\mathsf{E}}^*_k (r)) > a$. Again, studying the recursive definition of $\denote{\mathsf{E}_{\Box}}_{(-)}$ we observe that $\denote{\mathsf{E}_{\Box}}_{m+k}(\mu r) \geq \denote{\mathsf{E}}_m(\denote{\mathsf{E}}^*_k (r'))$, so we conclude that $\denote{\mathsf{E}_{\Box}}(\mu r) > a$. 
This is for all such $a \in \Quanti$, so $\denote{\mathsf{E}_{\Box}}(\mu r) = \denote{\mathsf{E}_{\Box}}(\denote{\mathsf{E}_{\Box}}^*(r))$.
\end{proof}

We end this section with an example of an equivalence and an in-equivalence. It has to be said that the purpose of this paper is to give a widely applicable approach to \emph{defining} equivalence, not to prove equivalence of terms. Moreover, for practical purposes, establishing an in-equivalence is easier than establishing an equivalence, since you only have to find one formula which distinguishes the two.

As an example, we look at the combination of cost and nondeterminism, given by signature $\Sigma_{cn}$, truth space $[0,\infty]$ with reverse order, and modalities $\{\mathsf{C}_{\lozenge}, \mathsf{C}_{\Box}\}$.
Consider the two terms $\ComTerm := \EffCost{1}{\return{\numeral{7}}}$ and $\ComTermT := \textsf{nor}(\return{\numeral{7}}, \EffCost{3}{\return{\numeral{7}}})$, both of which always return $7$, though $\ComTermT$ may return it at a lower or higher cost.
Then $\ComTerm \not\PrePos \ComTermT$, since $\Qmodels{\ComTerm}{\ModalForm{\mathsf{C}_{\lozenge}}{\EquForm{7}}} = 1 \not\leq 0 = \Qmodels{\ComTermT}{\ModalForm{\mathsf{C}_{\lozenge}}{\EquForm{7}}}$, and $\ComTermT \not\PrePos \ComTerm$, since $\Qmodels{\ComTermT}{\ModalForm{\mathsf{C}_{\Box}}{\EquForm{7}}} = 3 \not\leq 1 = \Qmodels{\ComTerm}{\ModalForm{\mathsf{C}_{\Box}}{\EquForm{7}}}$.
However, taking $\ComTermD := \textsf{nor}(\ComTerm, \ComTermT)$, it can be shown that $\ComTermD \equiv \ComTermT$, since for any formula $\phi \in \formulas{\NatType}$, $\Qmodels{\ComTermT}{\ModalForm{\mathsf{C}_{\lozenge}}{\phi}} = \Qmodels{\numeral{7}}{\phi} = \Qmodels{\ComTermD}{\ModalForm{\mathsf{C}_{\lozenge}}{\phi}}$ and $\Qmodels{\ComTermT}{\ModalForm{\mathsf{C}_{\Box}}{\phi}} = \Qmodels{\numeral{7}}{\phi} + 3 = \Qmodels{\ComTermD}{\ModalForm{\mathsf{C}_{\Box}}{\phi}}$.

\section{Applicative Bisimilarity}\label{section:bisimilarity}
We investigate how our quantitative modalities can be used to define a notion of Abramsky's applicative bisimilarity \cite{Abramsky90}, related to the behavioural equivalence (Theorem \ref{theorem:logic_bisim}), starting off by defining a relator \cite{Levy11,Thijs}.

\begin{definition}
	Given $\QObser$, we have a \emph{relator} $\Qrelator{-}: \mathcal{P}(X \times Y) \rightarrow \mathcal{P}(TX \times TY)$ given by:
	$$t \,\Qrelator{\mathcal{R}}\, r \iff \forall q \in \QObser, \forall h:X \rightarrow \Quanti, t \in q(h) \fimp r \in q(\RightSet{\mathcal{R}}{h})$$
\end{definition}

\noindent
We write $x \mathcal{R} y$ for $(x,y) \in \mathcal{R}$. Remember from the previous section that $(t \in q(h)) = \denote{q}(t[x \mapsto h(x)])$ and $(\RightSet{\mathcal{R}}{h})(b) := \sup\{h(a) \mid a \in X,a \mathcal{R} b\}$.
Note that $\Qrelator{\fimp} = \,\prebas$ and $\Qrelator{\prebas} = \,\modbel$ (see Lemma \ref{lem:quanmodbel_equivalences}). The following characterisation of the relator is immediate:
\begin{lemma}\label{lem:alt_rel}
	For any $t \in TX$, $r \in TY$, and $\mathcal{R} \subseteq X \times Y$, \, $t \,\Qrelator{\mathcal{R}}\, r \, \iff \, \forall h:X \rightarrow \Quanti, h^*(t) \prebas (\RightSet{\mathcal{R}}{h})^*(r)$.
\end{lemma}
The following lemma shows that this satisfies the usual properties of monotone relators from \cite{Levy11,Thijs}. The proof is technical yet straightforward, and is left out to preserve space. 

\begin{lemma}\label{lemma:relator}
	If all quantitative modalities from $\QObser$ are leaf-monotone, then $\Qrelator{-}$ has the following properties:
	\begin{enumerate}
		\item If $\mathcal{R}$ is reflexive, then so is $\Qrelator{\mathcal{R}}$.
		\item If $\mathcal{R} \subseteq \mathcal{S}$, then $\Qrelator{\mathcal{R}} \subseteq \Qrelator{\mathcal{S}}$.
		\item For $\mathcal{R} \subseteq X \times Y$ and $\mathcal{S} \subseteq Y \times Z$, $\Qrelator{\mathcal{R}}\Qrelator{\mathcal{S}} \subseteq \Qrelator{\mathcal{RS}}$. Here $\mathcal{R}\mathcal{S}$ is relational concatenation.
		\item $\forall f: X \to Z, g: Y \to W, \mathcal{R} \subseteq Z \times W$ it holds that $\Qrelator{(f \times g)^{-1}\mathcal{R}} = (f^* \times g^*)^{-1} \Qrelator{\mathcal{R}}$
		
		where $(f \times g)^{-1}(\mathcal{R}) = \{(x,y) \in X \times Y \,|\, f(x)\, \mathcal{R} \,g(y)\}$.
	\end{enumerate}
\end{lemma}

Fundamental to the definition of the relator is the notion of the \emph{right-predicate} $\RightSet{\mathcal{R}}{h}$. When the relation in question is our behavioural preorder, these right-predicates can be expressed in the logic.

\begin{lemma}\label{lemma:rightset_formula}
	Take $\mathcal{L} \in \{\GenLog,\PosLog\}$. For any predicate $\SetVar : \Terms{\ValType} \to \Quanti$, there is a formula $\phi^D \in \mathcal{L}$ such that $\Qmodels{\ValTerm}{\phi^D} = (\RightSet{\LogPre{\mathcal{L}}}{D})(\ValTerm)$ for any term $\ValTerm \in \Terms{\ValType}$.
\end{lemma}
\begin{proof}
	We use Lemma \ref{lemma:char_formula} to define $\phi^D := \DisForm\{\ConForm\{\StabForm{D(\ValTerm)}, \chi^{\ValTerm}\} \mid \ValTerm \in \Trees{\ValType} \}$.
\end{proof} 

In the case that $\mathcal{R}$ is a relation on terms of some value type $\ValType$, we write $\Qrelator{\mathcal{R}}$ for the relation on terms of type $\ForceType{\ValType}$ given by $\Qrelator{\{(\return{\ValTerm},\return{\ValTermT}) \mid \ValTerm \mathcal{R_{\ValType}} \ValTermT\}}$. 
A relation $\mathcal{R}$ on terms is \emph{well-typed}, if it only relates terms of the same type and context, and $\mathcal{R}$ is \emph{closed} if it only relates closed terms.

\begin{definition}\label{def:simulation}
	A well-typed closed relation $\mathcal{R}$ is an \emph{applicative $\QObser$-simulation} if:
	\begin{multicols}{2}
	\begin{enumerate}
		\item $\ValTerm \mathcal{R}_{\NatType} \ValTermT \implies V=W$.
		\item $\ComTerm \mathcal{R}_{\ThunkType{\ComType}} \ComTermT \implies \force{\ComTerm} \,\mathcal{R}_{\ComType}\, \force{\ComTermT}$.
		\item $\pair{j}{\ValTerm} \mathcal{R}_{\SumType{i}{I}{\ValType_i}} \pair{k}{\ValTermT} \implies (j = k) \wedge \ValTerm \mathcal{R}_{\ValType_j} \ValTermT$.
		\item $\pair{\ValTerm}{\ValTerm'} \mathcal{R}_{\PairType{\ValType}{\ValTypeT}} \pair{\ValTermT}{\ValTermT'} \implies \ValTerm \mathcal{R}_{\ValType} \ValTermT \wedge \ValTerm' \mathcal{R}_{\ValTypeT} \ValTermT'$.
		\item $\ComTerm \mathcal{R}_{\FunctType{\ValType}{\ComType}} \ComTermT \implies \forall V : \ValType, \,\,\appt{\ComTerm}{\ValTerm} \,\mathcal{R}_{\ComType}\, \appt{\ComTermT}{\ValTerm}$. 
		\item $\ComTerm \mathcal{R}_{\ProdType{i}{I}{\ComType_i}} \ComTermT \implies \forall j \in I, \,\, \appt{\ComTerm}{j} \,\mathcal{R}_{\ComTerm_j}\, \appt{\ComTermT}{j}$.
		\item $\ComTerm \mathcal{R}_{\ForceType{\ValType}} \ComTermT \implies |\ComTerm| \,\Qrelator{\mathcal{R}_{\ValType}}\, |\ComTermT|$ .
	\end{enumerate}
	\end{multicols}
	The \emph{applicative $\QObser$-similarity} is the largest applicative $\QObser$-simulation, whereas the \emph{applicative $\QObser$-bisimilarity} is the largest \emph{symmetric} applicative $\QObser$-simulation.
\end{definition}

\begin{theorem}\label{theorem:logic_simil}
	If all quantitative modalities from $\QObser$ are leaf-monotone, then the positive behavioural preorder $\PrePos$ is the applicative $\QObser$-similarity.
\end{theorem}

\begin{proof}
	Note that $\PrePos$ satisfies the first 6 properties for being a $\QObser$-simulation as a consequence of Lemma \ref{lem:pre_char}.
	We prove the seventh property:
	
	Assume $\ComTerm \PrePos \ComTermT$, $\qmodal \in \QObser$ and $\SetVar : \Terms{\ValType} \to \Quanti$. 
	We use Lemma \ref{lemma:rightset_formula} to find a formula $\phi^D$ such that $\phi^D(\ValTerm) = (\RightSet{\PrePos}{D})(\ValTerm)$.
	By reflexivity of $\PrePos$, we have $D(\ValTerm) \fimp (\RightSet{\PrePos}{D})(\ValTerm)$, so by leaf-monotonicity and $\ComTerm \PrePos \ComTermT$ it holds that: $\denote{\qmodal}(D^*(|\ComTerm|)) \fimp \Qmodels{\ComTerm}{\ModalForm{\qmodal}{\phi^D}} \fimp \Qmodels{\ComTermT}{\ModalForm{\qmodal}{\phi^D}} = \denote{\qmodal}((\RightSet{\PrePos}{D})^*|\ComTermT|)$.
	We can conclude that $|\ComTerm| \Qrelator{\PrePos_{\ValType}} |\ComTermT|$.
	So we proved that $\PrePos$ is a $\QObser$-simulation.
	
	We now need to prove that $\PrePos$ contains any other $\QObser$-simulation $\mathcal{R}$.
	To do that, we show that $\mathcal{R}$ preserves any formula $\gphi$ in the following sense: If $\GenTerm \,\mathcal{R}\, \GenTermT$, then $\Qmodels{\GenTerm}{\gphi} \fimp \Qmodels{\GenTermT}{\gphi}$.
	We do this by induction on formulas, using the fact that any formula is well-founded. 
	Assume $\GenTerm \,\mathcal{R}\, \GenTermT$.
	Suppose $\mathcal{R}$ preserves any formula from $X \subseteq \formulas{\GenTerm}$. Then $\Qmodels{\GenTerm}{\bigvee X} = \sup\{\Qmodels{\GenTerm}{\gphi} \mid \gphi \in X\} \fimp \sup\{\Qmodels{\GenTermT}{\gphi} \mid \gphi \in X\} = \Qmodels{\GenTermT}{\bigvee Y}$ and $\Qmodels{\GenTerm}{\bigwedge X} = \inf\{\Qmodels{\GenTerm}{\gphi} \mid \gphi \in X\} \fimp {\inf\{\Qmodels{\GenTermT}{\gphi} \mid \gphi \in X\}} = \Qmodels{\GenTermT}{\bigwedge Y}$.
	Assume $\mathcal{R}$ preserves formula $\gphi$, and $\Qmodels{\GenTerm}{\StepForm{\gphi}{a}} \neq \False$, then $\Qmodels{\GenTerm}{\gphi} \fpmi a$, so $\Qmodels{\GenTermT}{\gphi} \fpmi a$, hence $\Qmodels{\GenTermT}{\StepForm{\gphi}{a}} = \True$.
	$\mathcal{R}$ obviously preserves constant formulas, since $\Qmodels{\GenTerm}{\StabForm{a}} = a \fimp a = \Qmodels{\GenTermT}{\StabForm{a}}$.
	
	It is not difficult to prove that $\mathcal{R}$ preserves most basic formulas. The only difficult formula to consider is $\ModalForm{q}{\phi} \in \formulas{\ForceType{\ValType}}$. Assume $\ComTerm \mathcal{R} \ComTermT$, so by simulation property $|\ComTerm| \,\Qrelator{\mathcal{R}}\, |\ComTermT|$. By induction hypothesis and relator property 2 in Lemma \ref{lemma:relator}, it holds that $|\ComTerm| \Qrelator{\PrePos_{\ValType}} |\ComTermT|$. 
	Hence $ \Qmodels{\ComTerm}{\ModalForm{\qmodal}{\phi}} \fimp {\denote{\qmodal}((\RightSet{\PrePos}{\ValTerm \mapsto \Qmodels{\ValTerm}{\phi}})^*(|\ComTermT|))}$,
	which since $(\RightSet{\PrePos}{\ValTerm \mapsto \Qmodels{\ValTerm}{\phi}})(\ValTermT) \fimp \Qmodels{\ValTermT}{\phi}$ means with leaf-monotonicity of $\qmodal$ that $\Qmodels{\ComTerm}{\ModalForm{\qmodal}{\phi}} = \denote{\qmodal}(|\ComTerm|[\phi]) \fimp \denote{\qmodal}(|\ComTermT|[\phi]) = \Qmodels{\ComTermT}{\ModalForm{\qmodal}{\phi}}$. We conclude that $\ComTerm \PrePos \ComTermT$. 

	We can conclude that $\PrePos$ is the largest $\QObser$-simulation, hence it is equal to $\QObser$-similarity.
\end{proof}

Note the crucial use of Lemma \ref{lemma:rightset_formula} in the proof, which explains the need of the step-formulas in the logic.

\begin{theorem}\label{theorem:logic_bisim}
	If all quantitative modalities from $\QObser$ are leaf-monotone, then the general behavioural preorder $\PreGen$ is the largest symmetric applicative $\QObser$-simulation, and hence equal to applicative $\QObser$-bisimilarity.
\end{theorem}
\begin{proof}
	Firstly, it holds by Lemma \ref{lemma:symmetric} that $\PreGen$ is symmetric.
	Secondly, $\PreGen$ is a $\QObser$-simulation by the same proof as above. Lastly, any symmetric $\QObser$-simulation $\mathcal{R}$ is included in $\PreGen$, using a similar proof as above, proving with induction on formulas $\gphi$ that $\GenTerm \,\mathcal{R}\, \GenTermT$ implies $\Qmodels{\GenTermT}{\gphi} = \Qmodels{\GenTerm}{\gphi}$.
\end{proof}

\subsection{Howe's method}\label{section:Howe}
In this subsection, we briefly outline how the Howe's method \cite{How89,How} can be used to establish compatibility for the open extension of applicative $\QObser$-similarity and $\QObser$-bisimilarity as in \cite{Relational,modal}.
Firstly, we need some properties of the relators in addition to Lemma \ref{lemma:relator}. The proofs are technical and are left out to preserve space.

\begin{lemma}\label{relat:order}
	If all $q \in \QObser$ are leaf-monotone and Scott tree continuous, then the following four properties hold:
	\begin{enumerate}
		\item $t \QObser(\mathcal{R}) r \wedge t' \Tord t \wedge r \Tord r' \implies t' \QObser(\mathcal{R}) r'$.
		\item for any chain of trees $t_0 \Tord t_1 \Tord t_2 \Tord \dots$,
		$\forall n (t_n \QObser(\mathcal{R}) r_n) \Rightarrow (\sqcup_n t_n) \QObser(\mathcal{R}) (\sqcup_n r_n)$.
		\item $\mathcal{R} \subseteq X \times Y$, $\mathcal{S} \subseteq A \times B$ and $(f,g): X \times Y \to A \times B$, $(x \mathcal{R} y \Rightarrow f(x) \mathcal{S} g(y)) \implies (t \QObser(\mathcal{R}) r \Rightarrow f^*(t)  \QObser(\mathcal{S}) g^*(r))$.
		\item $x \mathcal{R} y \implies \eta(x) \Qrelator{\mathcal{R}} \eta(y)$
	\end{enumerate}
\end{lemma}	

The following lemma necessarily uses the property of decomposability.

\begin{lemma}\label{relat:monad}
	Given a decomposable set $\QObser$ of leaf-monotone modalities, then: $a \QObser(\QObser(\mathcal{R})) b \Rightarrow \mu a \QObser(\mathcal{R}) \mu b$.
\end{lemma}

\begin{proof}
	Assume $a \,\QObser(\QObser(\mathcal{R}))\, b$ and take some $h: X \to \Quanti$. 
	Let $t := h^{**}(a)$ and $r := {(\RightSet{\mathcal{R}}{h})^{**}(b)}$, then $\mu t = h^*(\mu a)$ and  $\mu r = {\RightSet{\mathcal{R}}{h}^*(\mu b)}$.
	We prove that $t \modbel r$ using Lemma \ref{lem:quanmodbel_equivalences}, so we can use decomposability to derive $\mu t \prebas \mu r$, from which we can conclude that $h^*(\mu a) \prebas (\RightSet{\mathcal{R}}{h})^*(\mu b)$, which by Lemma \ref{lem:alt_rel} implies $\mu t \QObser(\mathcal{R}) \mu r$.
	
	Let $q \in \QObser$ and $H \in \QBS$, then $(t \in q(H)) = \denote{q}(H^*(t)) = \denote{q}((H \circ h^*)^*(a)) = (a \in q(H \circ h^*)) \fimp {(b \in q(\RightSet{\QObser(\mathcal{R})}{H \circ h^*}))}$.
	Now $(\RightSet{\QObser(\mathcal{R})}{H \circ h^*})(k) = \sup\{H(h^*(l)) \mid l \QObser(\mathcal{R}) k\} \fimp \sup\{H(h^*(l)) \mid h^*(l) \prebas (\RightSet{\mathcal{R}}{h})^*(k) \}$ by Lemma \ref{lem:alt_rel}.
	Since $H \in \QBS$, the statement is smaller than $H((\RightSet{\mathcal{R}}{h})^*(k))$. 
	Hence using leaf-monotonicity, ${(b \in q(\RightSet{\QObser(\mathcal{R})}{(H \circ h^*)}))} \fimp {(b \in q(H \circ (\RightSet{\mathcal{R}}{h})^*)} = (r \in q(H))$.
	So we can conclude that $(t \in q(H)) \fimp (r \in q(H))$, and we are finished.
\end{proof}

As a consequence of the above lemmas, the following holds.

\begin{corollary}\label{relat:preserve}
	Given that $\QObser$ is a decomposable set of leaf-monotone and Scott tree continuous modalities:
	\begin{enumerate}
		\item $(\forall x,y, x \mathcal{R} y \Rightarrow f(x) \QObser(\mathcal{S}) g(y)) \wedge t \QObser(\mathcal{R}) r \Rightarrow \mu (f^*(t)) \QObser(\mathcal{S}) \mu (g^*(r))$
		\item $(\forall k, u_k \QObser(\mathcal{S}) v_k) \Rightarrow \EfOp(u_0,u_1,...) \QObser(\mathcal{S}) \EfOp(v_0,v_1,...)$
	\end{enumerate}
\end{corollary}

One of the contributions of this paper is identifying the properties on quantitative modalities for which the above relator properties are satisfied, such that we can apply Howe's method.
The application of Howe's method itself is however not novel, and is simply an alteration of the proof used for the call by value case in \cite{Relational,modal} (untyped and simply-typed respectively), using results from \cite{Lassen}.
As such, details of the proof have been omitted.
In short, Howe's method allows us to establish the following theorem.

\begin{theorem}\label{theorem:bisim}
	If $\QObser$ is a decomposable set of leaf-monotone and Scott tree continuous quantitative modalities, then $\QObser$-similarity and $\QObser$-bisimilarity are compatible.
\end{theorem}

Combining theorems \ref{theorem:logic_simil}, \ref{theorem:logic_bisim}, and \ref{theorem:bisim} we can derive Theorem \ref{main_theorem}, that the general and positive behavioural equivalence/preorder are compatible.

\section{Discussions}\label{section:conclusion}
We have generalised the logic from \cite{modal} to a \emph{quantitative} logic for terms of a call-by-push-value language with general recursion and several (combinations of) algebraic effects.
The quantitative logic is expressive, contains only meaningful behavioural properties, and induces a compatible program equivalence on terms.

In this paper, we consider program properties (or observations) as the primary way of describing program behaviour.
According to this philosophy, the generalisation to quantitative properties is natural.
Alternatively, one could consider relations (or comparisons) as primary, and instead generalise to quantitative relations. The resulting theory is that of metrics, along the lines of \cite{ARNOLD1980,Escardo_metric,MardarePP16}.
Relating the logic from this paper, or a variation thereof, to metrics (e.g. like the ones in \cite{appdis}) is a topic for future research.

The quantitative logic does not however naturally induce a metric on the terms. This is mainly because of the inclusion of step-formulas $\StepForm{\phi}{a}$, which take the quantitative information from $\phi$ and collapses it to a binary value.
These step-formulas are necessary for relating the behavioural equivalence with the applicative bisimilarity.
Their necessity can be seen as a natural consequence of the non-linearity of the language.
E.g., in the case of probability with $\Quanti := [0,\infty]$, the step-formula can be constructed using products of formulas.

The quantitative logic is very expressive, 
allowing one to deal with some awkward combinations of effects that are not amenable to a boolean treatment. 
Despite the many examples of combination of effects, there is no general theory for quantitative modalities of combined effects.
Such a theory is a potential subject for further research.
It would also be interesting to look at other examples of effects which the quantitative logic could describe,
like a the algebraic jump effect described in \cite{jumps},
or some form of concurrency.

The logic and examples from \cite{modal} can be considered as further examples for this paper, where one considers $\Quanti := \{\True,\False\}$. 
The property of Scott openness is the Boolean version of a combination of Scott tree continuity and leaf-monotonicity, and the notion of decomposability is a quantitative generalisation of the notion from \cite{modal} with the same name.
It should be noted, however, that most modalities from \cite{modal} are not sequential.

Along the lines of \cite{modal}, it is possible to define a \emph{pure variation} of the logic. This is a logic independent of the term syntax, using function formulas of the form $\FunctForm{\psi}{\cphi}$ instead of $\FunctForm{\ValTerm}{\cphi}$, where $\Qmodels{\ComTerm}{(\FunctType{\psi}{\cphi})} := \bigwedge\{\Qmodels{\appt{\ComTerm}{\ValTerm}}{\cphi} \mid \Qmodels{\ValTerm}{\psi} = \True\}$.
The logical equivalence of this pure logic will be equal to the behavioural equivalence, if the behavioural equivalence is compatible.

The denotation $\denote{q}: T\Quanti \to \Quanti$ of quantitative modalities are, in the case of the running examples, Eilenberg-Moore algebras. These are algebras $a: TX \to X$ such that $a \circ \eta_X = \textit{id}_X$ and $a \circ Ta = a \circ \mu_X$, the second statement coincides with the property of sequentiality in this paper.
As such, our example modalities potentially fit into the framework of Hasuo \cite{Hasuo:2015}. 
Connections between the two approaches may be explored in the future.

Since the theory has been formulated for call-by-push-value, it is not difficult to extract logics for specific reduction strategies including; call-by-name, call-by-value and lazy PCF \cite{LevyThesis,CBPV}.
The language can also be extended with universal polymorphic and recursive types. 
These extensions of the language are worked out in the author's forthcoming thesis. Further extensions could also be considered in the future.

\paragraph{Acknowledgements:}
I would like to thank Alex Simpson, Francesco Gavazzo and Aliaume Lopez  for all the helpful discussions and comments.

\bibliographystyle{plain}
\bibliography{quan_push_bib}

\appendix

\section{Howe's method proof}
In this appendix we give the proofs for Howe's method outlined in Subsection \ref{section:Howe}, to establish the compatibility of
applicative $\QObser$-(bi)similarity, and hence of the behavioural preorders. We start with the proofs of the relator properties.

\begin{proof}[Proof of Lemma \ref{lemma:relator}]
		1. If $\mathcal{R}$ is reflexive, then $\forall x, \RightSet{\mathcal{R}}{h}(x) = \sup\{h(a) \mid a \mathcal{R} x\} \fpmi h(x)$ since $x \mathcal{R} x$, hence $(h^*(t)) \Dord \RightSet{\mathcal{R}}{h}^*(t)$. So by leaf-monotonicity, for any $\qmodal \in \QObser$, 
		$\denote{\qmodal}(h^*(t)) \fimp \denote{\qmodal}(\RightSet{\mathcal{R}}{h}^*(t))$.
		
		2. If $\mathcal{R} \subseteq \mathcal{S}$, then for any $x$ it holds that $\RightSet{\mathcal{R}}{h}(x) \fimp \RightSet{\mathcal{S}}{h}(x)$. So by leaf-monotonicity, $\denote{\qmodal}(\RightSet{\mathcal{R}}{h}^*(t)) \fimp  \denote{\qmodal}(\RightSet{\mathcal{S}}{h}^*(t))$.
		
		3. $\RightSet{\mathcal{RS}}{h}(z) = \sup\{h(x) \mid \exists x \in X,x \,\mathcal{R}\mathcal{S}\, z\}= \sup\{h(x) \mid \exists x \in X, \exists y \in Y,x \mathcal{R} y,y \mathcal{S}z\} = \RightSet{\mathcal{R}}{\RightSet{\mathcal{S}}{h}}(z)$, hence $\RightSet{\mathcal{RS}}{h}^*(t) = \RightSet{\mathcal{R}}{\RightSet{\mathcal{S}}{h}}^*(t)$.
		
		4. If $t \, \QObser((f \times g)^{-1}\mathcal{R}) \, r$ then for all $h : Z \to \Quanti$, $h^*(f^*(t)) = (h \circ f)^*(t) \Bord (\RightSet{((f \times g)^{-1} \mathcal{R})}{h \circ f})^*(r) \Bord ((\RightSet{\mathcal{R}}{h}) \circ g)^*(r) = (\RightSet{R}{h})(g^*(r))$, so $t \, (f^* \times g^*)^{-1}(\Qrelator{\mathcal{R}}) \, r$.
		
		If $t \, (f^* \times g^*)^{-1}(\Qrelator{\mathcal{R}}) \, r$ then $\forall h : X \to \Quanti$, $h^*(t) \Bord (z \mapsto \sup\{h(x) \mid f(x) = z\})^*(f^*(t)) \Bord (\RightSet{\mathcal{R}}{z \mapsto \sup\{h(x) \mid f(x) = z\}})^*(g^*(r)) = (y \mapsto \sup\{h(x) \mid x \in X, z \,\mathcal{R}\, g(y), f(x) = z\})^*(r) \Bord (\RightSet{((f \times g)^{-1} \mathcal{R})}{h})^*(r)$, hence we conclude that $t \, \QObser((f \times g)^{-1}\mathcal{R}) \, r$.
\end{proof}

\begin{proof}[Proof of Lemma \ref{relat:order}]
	1. Assume $t' \Tord t \QObser(\mathcal{R}) r \Tord r'$ and take $q \in \QObser$ and $h: X \to \Quanti$.
	Since $t' \Tord t$ and $r \Tord r'$, $h^*(t') \Tord h^*(t')$ and $\RightSet{\mathcal{R}}{h}^*(r) \Tord \RightSet{\mathcal{R}}{h}^*(r')$.
	So by Scott tree continuity of $q$ and $t \QObser(\mathcal{R}) r$, it holds that $(t' \in q(h)) \fimp (t \in q(h)) \fimp (r \in q(\RightSet{\mathcal{R}}{h})) \fimp (r' \in q(\RightSet{\mathcal{R}}{h}))$.
	
	2. Take some chain of trees $t_0 \Tord t_1 \Tord t_2 \Tord \dots$ and assume $\forall n, (t_n \QObser(\mathcal{R}) r_n)$. 
	Let $q \in \QObser$ and $h: X \to \Quanti$, then by Scott tree continuity, $(\sqcup_n t_n \in q(h)) = \denote{q}(h^*(\sqcup_n t_n)) = \denote{q}(\sqcup_n h^*(t_n)) = \sup(\denote{q}(h^*(t_n)) \mid n \in \mathbb{N}) = \sup(t_n \in q(h) \mid n \in \mathbb{N}) \fimp\, \sup(r_n \in q(\RightSet{\mathcal{R}}{h}) \mid n \in \mathbb{N})$.
	By the same reasoning in reverse, the last statement is equal to $(\sqcup_n t_n \in q(\RightSet{\mathcal{R}}{h}))$, so we are finished.

	3. Assume $\mathcal{R}, \mathcal{S}, f$ and $g$ as above, and let $t,r$ be such that $t \QObser(\mathcal{R}) r$. 
	Take $h: A \to \Quanti$ and $q \in \QObser$, then $h \circ f: X \to \Quanti$, so $(f^*(t) \in q(h)) = (t \in q(h \circ f)) \fimp (r \in q(\RightSet{\mathcal{R}}{h \circ f})^*(r))$. 
	Now we have that: 
	$(\RightSet{\mathcal{R}}{h \circ f})(y) = \sup\{h(f(x)) \mid x \in X,x \mathcal{R} y\} \fimp \sup\{h(f(x)) \mid x \in X,f(x) \,\mathcal{S}\, g(y)\} \fimp \sup\{h(a) \mid a \in A,a \,\mathcal{S}\, g(y)\} = \RightSet{S}{h}(g(y))$. 
	So  $(\RightSet{\mathcal{R}}{h \circ f})^*(r) \,\Dord\, (\RightSet{S}{h})^*(g^*(r))$ and we use leaf-monotonicity to conclude that $(f^*(t) \in q(h)) \fimp (r \in q(\RightSet{\mathcal{R}}{h \circ f})) \fimp (g^*(r) \in q(\RightSet{S}{h}))$.
	
	4. This follows from the fact that if $x \,\mathcal{R}\, y$ then $\RightSet{\mathcal{R}}{h}(y) \fpmi h(x)$, so we can use leaf-monotonicity. Now for the second property.
\end{proof}	

\begin{proof}[Proof of Corollary \ref{relat:preserve}]
	\begin{enumerate}
		\item Using point (iii) of Lemma \ref{relat:order} on the assumptions we get $f^*(t) \, \mathcal{Q}(\mathcal{Q}(\mathcal{S})) \, g^*(r)$. We can then apply Lemma \ref{relat:monad} to get the correct result.
		\item We apply the previous property to the following data; $t = r = \sigma(0,1,2,\dots) \in T\Nat$, $f(n) = u_n$, $g(n) = v_n$ and $\mathcal{R} = \textit{id}_{\Nat}$. The conclusion follows directly.
	\end{enumerate}
\end{proof}

\subsection{The Howe closure}
For any closed relation $\mathcal{R}$ we can define the open extension $\OpeExt{\mathcal{R}}$ as $\Gamma \vdash \GenTerm \OpeExt{\mathcal{R}} \GenTermT : \GenType \Leftrightarrow \forall \overrightarrow{\ValTerm} : \Gamma, \GenTerm[\overrightarrow{\ValTerm}] \,\mathcal{R}\, \GenTermT[\overrightarrow{\ValTerm}]$. We define two more closure operations.

\begin{definition}\label{def:How}
	Taking an \emph{open} relation $\mathcal{R}$, we define the \emph{compatible refinement} $\ComExt{\mathcal{R}}$ using the derivation rules in Figure \ref{fig:compatible_extension}. For a closed relation $\mathcal{R}$ we define the \emph{Howe closure} $\HowClo{\mathcal{R}}$ as the smallest open relation $\mathcal{S}$ closed under the rules:
	\[
	\frac{\Gamma \vdash \ValTerm \ComExt{\mathcal{S}}_v W \quad \quad \Gamma \vdash W \OpeExt{\mathcal{R}}_v L}{\Gamma \vdash \ValTerm  \mathcal{S}_v L} \textbf{(HV)} \quad \quad \frac{\Gamma \vdash \ComTerm \ComExt{\mathcal{S}}_c \ComTermT \quad \quad \Gamma \vdash \ComTermT \OpeExt{\mathcal{R}}_c K}{\Gamma \vdash \ComTerm  \mathcal{S}_c K} \textbf{(HC)}
	\]
\end{definition}

\begin{figure}
	{\small
		$$
		\derivation{ }
		{\Vjudge{\Context}{\unit \,\ComExt{\mathcal{R}}\, \unit}{\UnitType}}\mathbf{C1}
		\biggap
		\derivation{ }
		{\Vjudge{\Context}{\zero \,\ComExt{\mathcal{R}}\, \zero}{\NatType}}\mathbf{C2}
		\biggap
		\derivation{\Vjudge{\Context}{\ValTerm \,\mathcal{R}\, \ValTerm'}{\NatType}}
		{\Vjudge{\Context}{\succe{\ValTerm} \,\ComExt{\mathcal{R}}\, \succe{\ValTerm'}}{\NatType}}\mathbf{C3}$$
		$$\derivation{\Vjudge{\Context}{\ValTerm \,\mathcal{R}\, \ValTerm'}{\NatType} \biggap \Cjudge{\Context}{\ComTerm \,\mathcal{R}\, \ComTerm'}{\ComType} \biggap 
			\Cjudge{\ContextCat{\Context}{\termvar : \NatType}}{\ComTermT \,\mathcal{R}\, \ComTermT'}{\ComType}}
		{\Cjudge{\Context}{\caseN{\ValTerm}{\ComTerm}{\termvar}{\ComTermT} \,\ComExt{\mathcal{R}}\, \caseN{\ValTerm}{\ComTerm'}{\termvar}{\ComTermT'}}{\ComTypeT}}\mathbf{C4}
		\biggap
		\derivation{ }
		{\Vjudge{\ContextCat{\ContextCat{\Context}{\termvar : \ValType}}{\Context'}}{\termvar \,\ComExt{\mathcal{R}}\, \termvar}{\ValType}}\mathbf{C5}
		$$
		$$
		\derivation{\Vjudge{\Context}{\ValTerm \,\mathcal{R}\, \ValTerm'}{\ValType} \biggap \Cjudge{\ContextCat{\Context}{\termvar : \ValType}}{\ComTerm \,\mathcal{R}\, \ComTerm'}{\ComType}}
		{\Cjudge{\Context}{\valseq{\termvar}{\ValTerm}{\ComTerm} \,\ComExt{\mathcal{R}}\, \valseq{\termvar}{\ValTerm'}{\ComTerm'}}{\ComType}}\mathbf{C6}
		\biggap
		\derivation{\Vjudge{\Context}{\ValTerm \,\mathcal{R}\, \ValTerm'}{\ValType}}
		{\Vjudge{\Context}{\return{\ValTerm} \,\ComExt{\mathcal{R}}\, \return{\ValTerm'}}{\ForceType{\ValType}}}\mathbf{C7}
		$$
		$$
		\derivation{\Cjudge{\Context}{\ComTerm \,\mathcal{R}\, \ComTerm'}{\ForceType{\ValType}} \biggap \Cjudge{\ContextCat{\Context}{\termvar : \ValType}}{\ComTermT \,\mathcal{R}\, \ComTermT'}{\ComType}}
		{\Cjudge{\Context}{\comseq{\ComTerm}{\termvar}{\ComTermT} \,\ComExt{\mathcal{R}}\, \comseq{\ComTerm'}{\termvar}{\ComTermT'}}{\ComType}}\mathbf{C8}
		\biggap
		\derivation{\Cjudge{\Context}{\ComTerm \,\mathcal{R}\, \ComTerm'}{\ComType}}
		{\Vjudge{\Context}{\thunk{\ComTerm} \,\ComExt{\mathcal{R}}\, \thunk{\ComTerm'}}{\ThunkType{\ComType}}}\mathbf{C9}$$
		
		$$\derivation{\Vjudge{\Context}{\ValTerm \,\mathcal{R}\, \ValTerm'}{\ThunkType{\ComType}}}
		{\Cjudge{\Context}{\force{\ValTerm} \,\ComExt{\mathcal{R}}\, \force{\ValTerm'}}{\ComType}}\mathbf{C10}
		\biggap
		\derivation{\Cjudge{\ContextCat{\Context}{\termvar : \ValType}}{\ComTerm \,\mathcal{R}\, \ComTerm'}{\ComType}}
		{\Cjudge{\Context}{\lambt{\termvar}{\ComTerm} \,\ComExt{\mathcal{R}}\, \lambt{\termvar}{\ComTerm'}}{\FunctType{\ValType}{\ComType}}}\mathbf{C11}
		$$
		$$
		\derivation{\Vjudge{\Context}{\ValTerm \,\mathcal{R}\, \ValTerm'}{\ValType} \biggap \Cjudge{\Context}{\ComTerm \,\mathcal{R}\, \ComTerm'}{\FunctType{\ValType}{\ComType}}}
		{\Cjudge{\Context}{\appt{\ComTerm}{\ValTerm} \,\ComExt{\mathcal{R}}\, \appt{\ComTerm'}{\ValTerm'}}{\ComType}}\mathbf{C12}
		\biggap
		\derivation{\Vjudge{\Context}{\ValTerm \,\mathcal{R}\, \ValTerm'}{\ValType_j}}
		{\Vjudge{\Context}{\pair{j}{\ValTerm} \,\ComExt{\mathcal{R}}\, \pair{j}{\ValTerm'}}{\SumType{i}{I}{\ValType_i}}}\mathbf{C13}
		$$
		$$
		\derivation{\Vjudge{\Context}{\ValTerm \,\mathcal{R}\, \ValTerm'}{\SumType{i}{I}{\ValType_i}} \biggap (\Cjudge{\ContextCat{\Context}{\termvar : \ValType_i}}{\ComTerm_i \,\mathcal{R}\, \ComTerm'_i}{\ComType}) \ \ \text{for each} \ \ i \in I}
		{\Cjudge{\Context}{\patmat{\ValTerm}{\{\dots,(i.\termvar)\,.\,\ComTerm_i,\dots\}} \,\ComExt{\mathcal{R}}\, \patmat{\ValTerm'}{\{\dots,(i.\termvar)\,.\,\ComTerm'_i,\dots\}}}{\ComType}}\mathbf{C14}
		$$
		\vspace{2mm}
		$$
		\derivation{\Vjudge{\Context}{\ValTerm \,\mathcal{R}\, \ValTerm'}{\ValType} \biggap \Vjudge{\Context}{\ValTermT \,\mathcal{R}\, \ValTermT'}{\ValTypeT}}
		{\Vjudge{\Context}{\pair{\ValTerm}{\ValTermT} \,\ComExt{\mathcal{R}}\, \pair{\ValTerm'}{\ValTermT'}}{\PairType{\ValType}{\ValTypeT}}}\mathbf{C15}
		\biggap
		\derivation{\Vjudge{\Context}{\ValTerm \,\mathcal{R}\, \ValTerm'}{\PairType{\ValType}{\ValTypeT}} \biggap \Cjudge{\ContextCat{\ContextCat{\Context}{\termvar : \ValType}}{\termvart : \ValTypeT}}{\ComTerm \,\mathcal{R}\, \ComTerm'}{\ComType}}
		{\Cjudge{\Context}{\patmat{\ValTerm}{\pair{\termvar}{\termvart}\,.\,\ComTerm} \,\ComExt{\mathcal{R}}\, \patmat{\ValTerm'}{\pair{\termvar}{\termvart}\,.\,\ComTerm}}{\ComType}}\mathbf{C16}
		$$
		$$
		\derivation{(\Cjudge{\Context}{\ComTerm_i \,\mathcal{R}\, \ComTerm'_i}{\ComType_i}) \ \ \text{for each} \ \ i \in I}
		{\Cjudge{\Context}{\prodlam{\ComTerm_i}{i \in I} \,\ComExt{\mathcal{R}}\, \prodlam{\ComTerm'_i}{i \in I}}{\ProdType{i}{I}{\ComType_i}}}\mathbf{C17}
		\biggap
		\derivation{\Cjudge{\Context}{\ComTerm \,\mathcal{R}\, \ComTerm'}{\ProdType{i}{I}{\ComType_i}}}
		{\Cjudge{\Context}{\appt{\ComTerm}{j} \,\ComExt{\mathcal{R}}\, \appt{\ComTerm'}{j}}{\ComType_j}}\mathbf{C18}
		$$
		$$
		\derivation{\Cjudge{\Context}{\ComTerm \,\mathcal{R}\, \ComTerm'}{\FunctType{\ThunkType{\ComType}}{\ComType}}}
		{\Cjudge{\Context}{\fix{\ComTerm} \,\ComExt{\mathcal{R}}\, \fix{\ComTerm'}}{\ComType}}\mathbf{C19}
		\biggap
		\derivation{\forall 1 \leq i \leq n.(\Cjudge{\Context}
			{\ComTerm_i \, \mathcal{R} \, \ComTerm'_i}{\ComType}) \biggap
			\EfOp : \FunctType{\TypeVar^n}{\TypeVar}}
		{\Cjudge{\Context}{\EfOp (\ComTerm_1,\dots,\ComTerm_n) 				 	
				\, \ComExt{\mathcal{R}} \, 
				\EfOp (\ComTerm'_1,\dots,\ComTerm'_n)}{\ComType}}\mathbf{C20}
		$$
		$$
		\derivation{\Cjudge{\Context, x : \NatType}
			{\ComTerm \, \mathcal{R} \, \ComTerm'}{\ComType} \biggap
			\EfOp : \FunctType{\TypeVar^{\NatType}}{\TypeVar}}
		{\Cjudge{\Context}{\EfOp (x \,.\, \ComTerm) \ComExt{\mathcal{R}} 
				\EfOp (x \,.\, \ComTerm)}{\ComType}}\mathbf{C21}
		$$
		$$
		\derivation{\Vjudge{\Context}{\ValTerm \, \mathcal{R} \, \ValTerm'}{\NatType} \biggap \forall 1 \leq i \leq n.(\Cjudge{\Context}
			{\ComTerm_i \, \mathcal{R} \, \ComTerm'_i}{\ComType}) \biggap 
			\EfOp : \FunctType{\NatType \times \TypeVar^n}{\TypeVar}}
		{\Cjudge{\Context}{\EfOp (\ValTerm, \ComTerm_1,\dots,\ComTerm_n) \, \ComExt{\mathcal{R}} \, \EfOp (\ValTerm', \ComTerm'_1,\dots,\ComTerm'_n)}{\ComType}}\mathbf{C22}
		$$
	}
	\caption{Compatible extension}
	\label{fig:compatible_extension}
\end{figure}

Given this definition, a well-typed relation $\mathcal{R}$ is compatible if and only if $\ComExt{\mathcal{R}} \subseteq \mathcal{R}$. The Howe closure is also the least solution to the equation $\mathcal{S} = \ComExt{\mathcal{S}} \OpeExt{\mathcal{R}}$ and the least solution to the inclusion  $\ComExt{\mathcal{S}} \OpeExt{\mathcal{R}} \subseteq \mathcal{S}$. We look at some preliminary results, mostly from Lassen \cite{Lassen}:

\begin{lemma}\label{Hprop1}
	If $\mathcal{R}$ is reflexive, then:
	\begin{enumerate}
		\item $\HowClo{\mathcal{R}}$ is compatible, hence reflexive.
		\item $\OpeExt{\mathcal{R}} \subseteq \HowClo{\mathcal{R}}$.
		\item If $\,\Context, x:\ValTypeT \vdash \GenTerm \, \HowClo{\mathcal{R}} \GenTermT$ and $\Vjudge{\Context}{\ValTerm, \ValTermT}{\ValTypeT}$ such that $\,\Context \vdash \ValTerm \HowClo{\mathcal{R}} W$, then $\,\Context \vdash \GenTerm[\ValTerm/x] \HowClo{\mathcal{R}} \GenTermT[W/x]$
	\end{enumerate}
\end{lemma}
\begin{proof} We prove the properties separately.
	\begin{enumerate}
		\item Since $\mathcal{R}$ is reflexive, so is $\OpeExt{\mathcal{R}}$. Hence $\HowClo{\ComExt{\mathcal{R}}} = \HowClo{\ComExt{\mathcal{R}}} \textit{id} \subseteq \HowClo{\ComExt{\mathcal{R}}} \OpeExt{\mathcal{R}} = \HowClo{\mathcal{R}}$.
		
		\item Note that the compatible refinement of a reflexive relation is reflexive. Hence $\ComExt{\HowClo{\mathcal{R}}}$ is reflexive, since $\HowClo{\mathcal{R}}$ is. So $\OpeExt{\mathcal{R}} = \textit{id} ~ \OpeExt{\mathcal{R}} \subseteq \ComExt{\HowClo{\mathcal{R}}} \OpeExt{\mathcal{R}} = \HowClo{\mathcal{R}}$.
		
		\item This requires an induction on the shape of $\GenTerm$ (which may be a value or a computation). We leave out $\Context$ in the proof for ease of notation. If $\{x\} \vdash \GenTerm \HowClo{\mathcal{R}} \GenTermT$ then by \textbf{HC} or \textbf{HV} we have $\{x\} \vdash \GenTerm \ComExt{\HowClo{\mathcal{R}}} \GenTermD$ and $\{x\} \vdash \GenTermD \OpeExt{\mathcal{R}} \GenTermT$. So we know $\GenTermD[\ValTermT/x] \OpeExt{\mathcal{R}} \GenTermT[\ValTermT/x]$. We need to prove, $\GenTerm[\ValTerm/x] \ComExt{\HowClo{\mathcal{R}}} \GenTermD[\ValTermT/x]$. In each of the cases of $\GenTerm$, $\{x\} \vdash \GenTerm \ComExt{\HowClo{\mathcal{R}}} \GenTermD$ is derived from rule \textbf{Cn} for some number $n$. This rule has as its premise some sequence of relations $\GenTerm_i \HowClo{\mathcal{R}} \GenTermD_i$. By induction hypothesis we have $\GenTerm_i[\ValTerm/x] \HowClo{\mathcal{R}} \GenTermD_i[\ValTermT/x]$, this is also trivially true in the base cases $n \in \{1,2,5\}$ since then the sequence is empty. Using \textbf{Cn} we can then derive that $\GenTerm[\ValTerm/x] \ComExt{\HowClo{\mathcal{R}}} \GenTermD[\ValTermT/x]$. Hence by \textbf{HV} or \textbf{HC} we get $\GenTerm[\ValTerm/x] \HowClo{\mathcal{R}} \GenTermT[\ValTermT/x]$. One can verify that this argument works for each of the cases of \textbf{Cn}.
	\end{enumerate}
\end{proof}

We can also say something about composing the Howe closure with the original relation, which by Lemma \ref{lemma:relator} works well with our relator.
\begin{lemma}\label{Hprop2}
	If $\mathcal{R}$ is a preorder relation on closed terms, then we have:
	\begin{enumerate}
		\item If $\GenTerm \HowClo{\mathcal{R}} \GenTermT$ and $\GenTermT \OpeExt{\mathcal{R}} \GenTermD$, then $\GenTerm \HowClo{\mathcal{R}} C$.
		\item For closed terms $\ComTerm,\ComTermT,\ComTermD$ of type $\ForceType{\ValType}$ such that $|\ComTerm| \QObser(\HowClo{\mathcal{R}}) |\ComTermT|$ and $|\ComTermT| \QObser(\OpeExt{\mathcal{R}}) |\ComTermD|$ we have $|\ComTerm| \QObser(\HowClo{\mathcal{R}}) |\ComTermD|$.
	\end{enumerate}
\end{lemma}
\begin{proof} We proof the properties individually. 
	\begin{enumerate}
		\item We use that $\mathcal{R}$ is transitive, hence $\OpeExt{\mathcal{R}}$ is transitive meaning $\OpeExt{\mathcal{R}} \OpeExt{\mathcal{R}} \subseteq \OpeExt{\mathcal{R}}$. Hence with $\HowClo{\mathcal{R}} = \HowClo{\ComExt{\mathcal{R}}} \OpeExt{\mathcal{R}}$ we have $\HowClo{\mathcal{R}} \OpeExt{\mathcal{R}} = (\HowClo{\ComExt{\mathcal{R}}} \OpeExt{\mathcal{R}}) \OpeExt{\mathcal{R}} \subseteq \HowClo{\ComExt{\mathcal{R}}} \OpeExt{\mathcal{R}} = \HowClo{\mathcal{R}}$.
		
		\item This follows from applying property 2 of Lemma \ref{lemma:relator} to the previous statement.
	\end{enumerate}
\end{proof}

\subsection{The Howe closure of an applicative $\QObser$-simulation}\label{section:closure}
We look at the Howe closure of a $\QObser$-simulation preorder $\arbsim$. We assume that $\QObser$ is a decomposable set of leaf-monotone and Scott tree continuous modalities. The lemmas proven in the previous two subsections are satisfied, hence we know that $(\arbsim) \subseteq \HowClo{\arbsim}$ by Lemma \ref{Hprop1}. We prove that $\HowClo{\arbsim}$ is a $\QObser$-simulation by explicitly checking the seven conditions from Definition \ref{def:simulation}.

\begin{lemma}\label{Nequ}
	If for $\ValTerm, \ValTermT:\NatType$ we have $\ValTerm \HowClo{\arbsim} \ValTermT$, then $\ValTerm = \ValTermT$.
\end{lemma}

\begin{proof}
	Using the inductive definition of $\HowClo{\arbsim}$ there must be an $L:\NatType$ such that $\ValTerm \hat{\HowClo{\arbsim}} L$ and $L \arbsim \ValTermT$, the latter meaning $L = \ValTermT$ because of the simulation property. The fact that $\ValTerm \hat{\HowClo{\arbsim}} L$ must have come as a conclusion from either \textbf{C3} or \textbf{C4}. In the first case, $\ValTerm = \zero= L$ and hence $\ValTerm = L = \ValTermT$. In the second case, $\ValTerm = \succe{\ValTerm'}$ and $L = \succe{L'}$ with $\ValTerm' \HowClo{\arbsim} L'$, and the proof has been reduced to showing $\ValTerm' = L'$, since then $\ValTerm = \succe{\ValTerm'} = \succe{L'} = L = \ValTermT$. We do induction on the structure of $\ValTerm$, which cannot go on forever since $\ValTerm$ is a syntactically finite term. So eventually we get to $\zero$ and we can make a conclusion of the form $\ValTerm = n\succe{\ValTerm''} = n\succe{\zero} = n\succe{L''} = L = \ValTermT$ for some $n \in \NatType$. That concludes the proof.
\end{proof}

The following lemma is evident from the compatibility properties.
\begin{lemma}\label{lemma:How_simul_var}
	By compatibility of $\HowClo{\arbsim}$ it holds that:
	\begin{enumerate}
		\item For $\ValTerm \HowClo{\arbsim}_{\ThunkType{\ComType}} \ValTermT$ we have $\force{\ValTerm} \HowClo{\arbsim}_{\ComType} \force{\ValTermT}$.
		\item For $\ComTerm \HowClo{\arbsim}_{\FunctType{\ValType}{\ComType}} \ComTermT$ we have $\forall \ValTerm \in \Terms{\ValType}, \appt{\ComTerm}{\ValTerm} \HowClo{\arbsim}_{\ComType} \appt{\ComTermT}{\ValTerm}$.
		\item For $\ComTerm \HowClo{\arbsim}_{\ProdType{i}{I}{\ComType_i}} \ComTermT$ we have $\forall j \in I, \appt{\ComTerm}{j} \HowClo{\arbsim}_{\ComType_j} \appt{\ComTermT}{j}$.
	\end{enumerate}
\end{lemma}

We can easily prove two more simulation properties.
\begin{lemma}\label{lemma:How_simul_sum}
	If $\pair{j}{\ValTerm} \HowClo{\arbsim}_{\SumType{i}{I}{\ValType_i}} \pair{k}{\ValTermT}$ then $j=k$ and $\ValTerm \HowClo{\arbsim}_{\ValType_j} \ValTermT$.
\end{lemma}
\begin{proof}
	There is a pair $\pair{l}{\ValTermD}$ such that $\pair{j}{\ValTerm} \ComExt{\HowClo{\arbsim}_{\SumType{i}{I}{\ValType_i}}} \pair{l}{\ValTermD} \OpeExt{\arbsim} \pair{k}{\ValTermT}$. The latter implies $l=k$ and $\ValTermD \arbsim \ValTermT$ by simulation property. The former statement can only have come from compatible extension rule $\mathbf{C13}$, so $j=l$ and $\ValTerm \HowClo{\arbsim} \ValTermD$. We can now use Lemma \ref{Hprop2} to conclude that $\ValTerm \HowClo{\arbsim} \ValTermT$.
\end{proof}

\begin{lemma}\label{lemma:How_simul_pair}
	If $\pair{\ValTerm}{\ValTerm'} \HowClo{\arbsim}_{\PairType{\ValType}{\ValTypeT}} \pair{\ValTermT}{\ValTermT'}$ then $\ValTerm \HowClo{\arbsim}_{\ValType} \ValTermT$ and $\ValTerm' \HowClo{\arbsim}_{\ValTypeT} \ValTermT'$.
\end{lemma}
\begin{proof}
	There is a pair $\pair{\ValTermD}{\ValTermD'}$ such that $\pair{\ValTerm}{\ValTerm'} \ComExt{\HowClo{\arbsim}_{\PairType{\ValType}{\ValTypeT}}} \pair{\ValTermD}{\ValTermD'} \OpeExt{\arbsim} \pair{\ValTermT}{\ValTermT'}$. The latter implies $\ValTermD \arbsim \ValTermT$ and $\ValTermD' \arbsim \ValTermT'$ by simulation property. The former statement can only have come from compatible extension rule $\mathbf{C15}$, so $\ValTerm \HowClo{\arbsim} \ValTermD$ and $\ValTerm' \HowClo{\arbsim} \ValTermD'$. We can now use Lemma \ref{Hprop2} to conclude that $\ValTerm \HowClo{\arbsim} \ValTermT$ and $\ValTerm' \HowClo{\arbsim} \ValTermT'$.
\end{proof}

So all conditions except 6 of being a $\QObser$-simulation are satisfied. Condition 6. is the most difficult to prove and requires an induction on the reduction relation of terms.

It requires us to look at terms $\GenTerm$, $\GenTermT$ of type $\
\ForceType{\ValType}$ such that $\GenTerm \HowClo{\arbsim} \GenTermT$, and prove that $|\GenTerm| \Qrelator{\HowClo{\arbsim}} |\GenTermT|$. Using Lemma \ref{relat:order}, this can be reduced to asking $|\GenTerm|_n \Qrelator{\HowClo{\arbsim}} |\GenTermT|$ for all $n$. This allows us to do an induction on the denotation map $|\GenTerm|$. In general, one would look at the shape of $\GenTerm$ and see what it reduces to after one step, so one can use the induction hypothesis. This is a relatively straightforward investigation in the fine-grained call by value case.

For call-by-push-value, we have the problem that effects may occur in any computation type, which is particularly problematic when considering non-producer type.  Concretely, it may be that our $\GenTerm:\ForceType{\ValType}$ is of the form $\GenTerm = \appt{\GenTerm'}{\ValTerm}$ where $\GenTerm':\ValTypeT \to F\ValType$ and $\ValTerm : \ValTypeT$. Now, because $\GenTerm'$ is of a computation type, it may not be of the form $\lambt{x}{\ComTerm}$. This is problematic, as we still do not have any clue to what $\GenTerm$ might reduce to. To investigate that, we would require another case analysis on $\GenTerm'$, which results in a bureaucratic nightmare. We can say that the application case is \emph{uninformative}, and we need to continue doing case analysis until we find a term that is not of the form of an application, which we call \emph{informative}. 
\begin{definition}
	A stack $\Stack$ is called a \emph{frame} if it does not contain any $\comseq{(-)}{\termvar}{\ComTerm}$ parts.
	A computation term is \emph{uninformative} if it is of the form $\Stack\{\ComTerm\}$ where $\Stack$ is a non-empty frame. Else it is called \emph{informative}.
\end{definition}

\begin{lemma}\label{lemma:simul_stack}
	For any frame $\Stack$, if $\ComTerm \arbsim \ComTermT$ then $\Stack\{\ComTerm\} \arbsim \Stack\{\ComTermT\}$.
\end{lemma}
Doing structural induction on $\GenTerm$, we observe the following result.
\begin{lemma}\label{case_frame}
	Any computation term $\GenTerm$ is of the form $\Stack\{\GenTerm'\}$ where $\Stack$ is a frame and $\GenTerm'$ is an informative term.
\end{lemma}

\begin{definition}
	Two frames $\Stack$ and $\Stackt$ \emph{match} when the following statements hold.
	\begin{enumerate}
		\item If $\Stack = \StackEmpty$ then $\Stackt = \StackEmpty$.
		\item If $\Stack = \Stack' \circ \ValTerm$ then $\Stackt = \Stackt' \circ \ValTerm'$ where $\ValTerm \HowClo{\arbsim} \ValTerm'$ and $\Stack'$ matches with $\Stackt'$.
		\item If $\Stack = \Stack' \circ i$ then $\Stackt = \Stackt' \circ i$ where $\Stack'$ matches with $\Stackt'$.
	\end{enumerate}
\end{definition}

We have the following property.

\begin{lemma}\label{match_frame}
	If $\Stack\{\GenTerm'\} \HowClo{\arbsim} \GenTermT$ then there is a frame $\Stackt$ and a term $\GenTermT'$ such that; $\Stack$ matches $\Stackt$, $\GenTerm' \ComExt{\HowClo{\arbsim}} \GenTermT'$ and $\Stackt\{\GenTermT'\} \arbsim \GenTermT$. 
\end{lemma}	
\begin{proof}
	We do this by induction on the frame. If $\Stack = \StackEmpty$, then the statements hold by taking $\Stackt = \StackEmpty$ and $\GenTermT'$ such that $\GenTerm' \ComExt{\HowClo{\arbsim}} \GenTermT' \arbsim \GenTermT$, which is possible since $(\HowClo{\arbsim}) = (\ComExt{\HowClo{\arbsim}}) \circ (\arbsim)$. Now for the induction step, assume the statement holds for any smaller frames $\Stack'$.
	\begin{enumerate}
		\item $\Stack = \ValTerm \circ \Stack'$, then $\Stack\{\GenTerm'\} = (\appt{\Stack'\{\GenTerm'\}}{\ValTerm})$. Now, there is a term $\GenTermD$ such that $\Stack\{\GenTerm'\} \ComExt{\HowClo{\arbsim}} \GenTermD \arbsim \GenTermT$. The statement $(\appt{\Stack'\{\GenTerm'\}}{\ValTerm}) \ComExt{\HowClo{\arbsim}} \GenTermD$ could  only have been derived from rule $\mathbf{C12}$, so we know there are $\GenTermD'$ and $\ValTermT$ such that $\GenTermD = (\appt{\GenTermD'}{\ValTermT})$, $\Stack'\{\GenTerm'\} \HowClo{\arbsim} \GenTermD'$, and $\ValTerm \HowClo{\arbsim} \ValTermT$.
		
		We use induction hypothesis on $\Stack'\{\GenTerm'\} \HowClo{\arbsim} \GenTermD'$ to find a term $\GenTermT'$ and frame $\Stackt'$ such that $\Stack'$ matches $\Stackt'$, $\GenTerm' \ComExt{\HowClo{\arbsim}} \GenTermT'$ and $\Stackt'\{\GenTermT'\} \arbsim \GenTermD'$.
		Let $\Stackt := \ValTermT \circ \Stackt'$, then $\Stack$ and $\Stackt$ match. From $\Stackt'\{\GenTermT'\} \arbsim \GenTermD'$ and simulation rule we have $\Stackt\{\GenTermT'\} = (\appt{\Stackt'\{\GenTermT'\}}{\ValTermT}) \arbsim (\appt{\GenTermD'}{\ValTermT}) = \GenTermD$. With $\GenTermD \arbsim \GenTermT$ we can conclude that $\Stackt\{\GenTermT'\} \arbsim \GenTermT$ and from earlier $\GenTerm' \ComExt{\HowClo{\arbsim}} \GenTermT'$. So $\Stackt$ and $\GenTermT'$ have the desired properties.

		\item $\Stack = i \circ \Stack'$, then $\Stack\{\GenTerm'\} = (\appt{\Stack'\{\GenTerm'\}}{i})$. Now, there is a term $\GenTermD$ such that $\Stack\{\GenTerm'\} \ComExt{\HowClo{\arbsim}} \GenTermD \arbsim \GenTermT$. The statement $(\appt{\Stack'\{\GenTerm'\}}{i}) \ComExt{\HowClo{\arbsim}} \GenTermD$ could  only have been derived from rule $\mathbf{C18}$, so we know there is a $\GenTermD'$ such that $\GenTermD = (\appt{\GenTermD'}{i})$ and $\Stack'\{\GenTerm'\} \HowClo{\arbsim} \GenTermD'$.
		
		We use induction hypothesis on $\Stack'\{\GenTerm'\} \HowClo{\arbsim} \GenTermD'$ to find a term $\GenTermT'$ and frame $\Stackt'$ such that $\Stack'$ matches $\Stackt'$, $\GenTerm' \ComExt{\HowClo{\arbsim}} \GenTermT'$ and $\Stackt'\{\GenTermT'\} \arbsim \GenTermD'$.
		Let $\Stackt := i \circ \Stackt'$, then $\Stack$ and $\Stackt$ match. From $\Stackt'\{\GenTermT'\} \arbsim \GenTermD'$ and simulation rule we have $\Stackt\{\GenTermT'\} = (\appt{\Stackt'\{\GenTermT'\}}{i}) \arbsim (\appt{\GenTermD'}{i}) = \GenTermD$. With $\GenTermD \arbsim \GenTermT$ we can conclude that $\Stackt\{\GenTermT'\} \arbsim \GenTermT$ and from earlier $\GenTerm' \ComExt{\HowClo{\arbsim}} \GenTermT'$. So $\Stackt$ and $\GenTermT'$ have the desired properties.
	\end{enumerate}
\end{proof}
Matching frames are very handy, since they can make use of compatibility:

\begin{lemma}
	If $\Stack$ matches $\Stackt$, then for all $\ComTerm \HowClo{\arbsim} \ComTermT$ we have $\Stack\{\ComTerm\} \HowClo{\arbsim} \Stackt\{\ComTermT\}$.
\end{lemma}

The last important property of frames is that it works nicely with respect to the reduction relation.

\begin{lemma}\label{lemma:frame_red}
	If $\reduct{\ComTerm}{\ComTermT}$, then $|\Stack\{\ComTerm\}| = |\Stack\{\ComTermT\}|$ and for any $k \in \Nat$ and any frame $\Stack$ we have $|\Stack\{\ComTerm\}|_{k+1} = |\Stack\{\ComTermT\}|_k$.
	
	\noindent
	Moreover, $|\Stack\{\EfOp(\ComTerm_0,\ComTerm_1,\dots)\}|_{n+1} = \EfOp\{|\Stack\{\ComTerm_0\}|_n,|\Stack\{\ComTerm_1\}|_n,\dots\}$ (same for other effects).
\end{lemma}

We have the necessary tools to prove the following lemma.

\begin{lemma}[Key Lemma]
	For any $n \in \Nat$ we have: $\GenTerm \HowClo{\arbsim} \GenTermT \implies |\GenTerm|_n \Qrelator{\HowClo{\arbsim}} |\GenTermT|$ for any two closed terms $\GenTerm$, $\GenTermT$ of type $\ForceType{\ValType}$.
\end{lemma}

\begin{proof}
	We do an induction on $n$. 
	
	\textbf{Base case}, $n = 0$. Here $|\ComTerm|_0 = \bot$ which is below any other tree. So $|\ComTerm|_0 \Qrelator{\HowClo{\arbsim}} |\ComTermT|$ can be derived by using both reflexivity of $\Qrelator{\HowClo{\arbsim}}$ (Lemma \ref{lemma:relator}) and Lemma \ref{relat:order}.
	
	\textbf{Induction step} $(n+1)$. We assume as the \emph{induction hypothesis} that for any $\GenTerm' \HowClo{\arbsim} \GenTermT'$ and $k \leq n$ we have $|\GenTerm'|_k \Qrelator{\HowClo{\arbsim}} |\GenTermT'|$. To prove, for any $\GenTerm \HowClo{\arbsim} \GenTermT$ we have $|\GenTerm|_{n+1} \Qrelator{\HowClo{\arbsim}} |\GenTermT|$. 
	
	\noindent
	Assume $\GenTerm \HowClo{\arbsim} \GenTermT$. 
	We use Lemma \ref{case_frame} to find a frame $\Stack$ and an informative term $\GenTerm'$ such that $\GenTerm = \Stack\{\GenTerm'\}$. 
	We then use Lemma \ref{match_frame} to find matching frame $\Stackt$ together with a term $\GenTermT'$ such that: 
	$\Stackt\{\GenTermT'\} \arbsim \GenTermT$ and $\GenTerm'\ComExt{\HowClo{\arbsim}} \GenTermT'$. 
	We want to prove that $|\GenTerm|_{n+1} = |\Stack\{\GenTerm'\}|_{n+1} \Qrelator{\HowClo{\arbsim}} |\Stackt\{\GenTermT'\}|$, since then with $\Stackt\{\GenTermT'\}  \arbsim \GenTermT$ we can conclude via Lemma \ref{Hprop2} and simulation property that $|\GenTerm|_{n+1} \Qrelator{\HowClo{\arbsim}} |\GenTermT|$.
	$$\text{To Prove: } |\GenTerm|_{n+1} = |\Stack\{\GenTerm'\}|_{n+1} \, \Qrelator{\HowClo{\arbsim}} \, |\Stackt\{\GenTermT'\}|$$
	We do a case distinction on $\GenTerm' : \ComType$, which is informative, so not of the form $\appt{\ComTerm}{\ValTerm}$ or $\appt{\ComTerm}{i}$.
	
	\noindent
	We start with the three unfold cases, where the $\Stack$ frame is actively used.
	
	\begin{enumerate}
		\item If $\GenTerm' = \return{\ValTerm} : \ComType$, which can only be of $\ForceType{}$-type, so the frame $\Stack$ must be $\StackEmpty$ as no other frames accept a term of this type. Hence $\GenTerm' = \GenTerm$ and $\ComType = \ForceType{\ValType}$. 
		The matching frame $\Stackt$ must be $\StackEmpty$ too, $|\GenTerm| = |\GenTerm'| = \eta(\ValTerm)$, and $\return{\ValTerm} \ComExt{\HowClo{\arbsim}} \GenTermT' = \Stackt\{\GenTermT'\}$. 
		This is only possible from the return compatibility rule $\mathbf{C7}$, meaning $\GenTermT' = \return{\ValTermT}$ for some $\ValTermT$ and $\ValTerm \HowClo{\arbsim} \ValTermT$. By Lemma \ref{relat:monad} we have $\eta(\ValTerm) \Qrelator{\HowClo{\arbsim}} \eta(\ValTermT)$, hence $|\Stack\{\GenTerm\}|_{n+1} = |\GenTerm|_{n+1} = |\return{\ValTerm}|_{n+1} = \eta(\ValTerm) \Qrelator{\HowClo{\arbsim}} \eta(\ValTermT) = |\return{\ValTermT}| = |\GenTermT'| = |\Stackt\{\GenTermT'\}|$.
		
		\item If $\GenTerm' = \lambt{x}{\ComTerm}$, then for $\Stack\{\GenTerm'\}$ to be of type $\ForceType{\ValType}$, $\Stack$ must be of the form $\Stack' \circ \ValTerm$. 
		Since $\Stack$ matches $\Stackt$, $\Stackt = \Stackt' \circ \ValTermT$ where $\Stack'$ matches $\Stackt'$ and $\ValTerm \HowClo{\arbsim} \ValTermT$. The statement $\GenTerm' = \lambt{x}{\ComTerm} \ComExt{\HowClo{\arbsim}} \GenTermT'$ could only have been derived via the lambda compatibility rule $\mathbf{C11}$, so $\GenTermT' = \lambt{x}{\ComTermT}$ for some $\ComTermT$ where $\Cjudge{\{x\}}{\ComTerm \HowClo{\arbsim} \ComTermT}{\ComTypeT}$. By Lemma \ref{Hprop1} we have that $\ComTerm[\ValTerm/x] \HowClo{\arbsim} \ComTermT[\ValTermT/x]$, so we can do the following derivation using the Lemma \ref{lemma:frame_red} and the induction hypothesis:
		
		$|\Stack\{\GenTerm'\}|_{n+1} = |\Stack'\{\appt{\lambt{x}{\ComTerm}}{\ValTerm}\}|_{n+1} = |\Stack'\{\ComTerm[\ValTerm/x]\}|_n \, \Qrelator{\HowClo{\arbsim}} \, |\Stackt'\{\ComTermT[\ValTermT/x]\}| = |\Stackt\{\GenTermT'\}|$.
		
		\item $\GenTerm' = \prodlam{\ComTerm_i}{i \in I}$, then for $\Stack\{\GenTerm'\}$ to be of type $\ForceType{\ValType}$, $\Stack$ must be of the form $\Stack' \circ j$. 
		Since $\Stack$ matches $\Stackt$, $\Stackt = \Stackt' \circ j$ where $\Stack'$ matches $\Stackt'$. The statement $\GenTerm' = \prodlam{\ComTerm_i}{i \in I} \ComExt{\HowClo{\arbsim}} \GenTermT'$ could only have been derived via the lambda compatibility rule $\mathbf{C17}$, so $\GenTermT' = \prodlam{\ComTermT_i}{i \in I}$ for some sequence of $\ComTermT_i$-s, where $\forall i. \ComTerm_i \HowClo{\arbsim} \ComTermT_i$. We can do the following derivation using Lemma \ref{lemma:frame_red} and the induction hypothesis:
		
		$|\Stack\{\GenTerm'\}|_{n+1} = |\Stack'\{\appt{\prodlam{\ComTerm_i}{i \in I}}{j}|_{n+1} = |\Stack'\{\ComTerm_j\}|_n \, \Qrelator{\HowClo{\arbsim}} \, |\Stackt'\{\ComTermT_j\}| = |\Stackt\{\GenTermT'\}|$.
		
	\enumeratext{Computation sequencing is special, and needs decomposability.}
		\item $\GenTerm' = \comseq{\ComTerm}{x}{\ComTermT}$, then $|\GenTerm'|_{n+1} \leq |\ComTerm|_n[\ValTerm \mapsto |\ComTermT[\ValTerm/x]|_n]$.
		Hence $|\GenTerm|_{n+1} \leq |\ComTerm|_n[\ValTerm \mapsto |\Stack\{\ComTermT[\ValTerm/x]\}|_n]$. Now $\GenTerm' \ComExt{\HowClo{\arbsim}} \GenTermT'$ results only from compatibility rule \textbf{C8}, hence $\GenTermT' = \comseq{\ComTerm'}{x}{\ComTermT'}$, $\ComTerm \HowClo{\arbsim} \ComTerm'$ and $x \vdash \ComTermT \HowClo{\arbsim} \ComTermT'$. 
		So $|\Stackt\{\GenTermT'\}| = |\ComTerm'|[\ValTermT \mapsto |\Stackt\{\ComTermT'[\ValTermT/x]\}|]$. 
		By induction hypothesis we have $|\ComTerm|_n \Qrelator{\HowClo{\arbsim}} |\ComTerm'|$ so define 
		
		$t := |\ComTerm|_n$ and $r = |\ComTerm'|$. 
		
		From $\ValTerm \HowClo{\arbsim} \ValTermT$ we get via Lemma \ref{Hprop1} that $\ComTermT[\ValTerm/x] \HowClo{\arbsim} \ComTermT'[\ValTermT/x]$, and since $\Stack$ and $\Stackt$ match we have $\Stack\{\ComTermT[\ValTerm/x]\} \HowClo{\arbsim} \Stackt\{\ComTermT'[\ValTermT/x]\}$ and hence by induction hypothesis we get $|\Stack\{\ComTermT[\ValTerm/x]\}|_n \Qrelator{\HowClo{\arbsim}} |\Stackt\{\ComTermT'[\ValTermT/x]\}|$. 
		So define 
		
		$f: \ValTerm \mapsto |\Stack\{\ComTermT[\ValTerm/x]\}|_n$ and $g: \ValTermT \mapsto |\Stackt\{\ComTermT'[\ValTermT/x]\}|$. 
		
		Now $t,r,f,g$ satisfy the conditions of Cor. \ref{relat:preserve}, so we conclude that:
		
		$|\GenTerm|_{n+1} \leq |\ComTerm|_n[\ValTerm \mapsto |\Stack\{\ComTermT[\ValTerm/x]\}|_n] \, \Qrelator{\HowClo{\arbsim}} \, |\ComTerm'|[\ValTermT \mapsto |\Stackt\{\ComTermT'[\ValTermT/x]\}|] = |\Stackt\{\GenTermT'\}|$.
		
	\enumeratext{Next come the reduction cases.}
		
		\item $\GenTerm' = \force{\ValTerm}$, 
		then since $\ValTerm$ is a closed value it must be of the form $\thunk{\ComTerm}$ for some term $\ComTerm$, 
		so $\reduct{\GenTerm'}{\ComTerm}$. 
		Now $\GenTerm' \ComExt{\HowClo{\arbsim}} \GenTermT'$ can only come from compatibility rule \textbf{C10}, 
		hence $\GenTermT' = \force{\ValTermT}$ where $\ValTerm \HowClo{\arbsim} \ValTermT$. 
		From \textbf{HV} we get $\thunk{\ComTerm} = \ValTerm \ComExt{\HowClo{\arbsim}} \ValTermD \arbsim \ValTermT$ for some $\ValTermD$, the relation can only come from \textbf{C9} so $\ValTermD = \thunk{\ComTermD}$ where $\ComTerm \HowClo{\arbsim} \ComTermD$.
		By simulation property of $\arbsim$ we get $\force{\thunk{\ComTermD}} \arbsim \force{\ValTermT} = \GenTermT'$.
		By Lemma \ref{lemma:simul_stack} we have $\Stackt\{\force{\thunk{\ComTermD}}\} \arbsim \Stackt\{\GenTermT'\}$, which then by $\ForceType{\ValType}$-simulation property implies that $|\Stackt\{\ComTermD\}| = |\Stackt\{\force{\thunk{\ComTermD}}\}| \Qrelator{\arbsim_{\ValType}} |\Stackt\{\GenTermT'\}|$.
		Using the induction hypothesis on $\Stack\{\ComTerm\} \HowClo{\arbsim} \Stackt\{\ComTermD\}$ we get $|\Stack\{\ComTerm\}|_n \QObser(\HowClo{\arbsim}) |\Stackt\{\ComTermD\}|$. Combining this with $|\Stackt\{\ComTermD\}| = |\Stackt\{\force{\thunk{\ComTermD}}\}| \Qrelator{\arbsim} |\Stackt\{\GenTermT'\}|$ using Lemma \ref{lemma:frame_red} and \ref{Hprop2} we get
		$|\GenTerm|_{n+1} = |\Stack\{\force{\thunk{\ComTerm}}\}|_{n+1} = |\Stack\{\ComTerm\}|_n \Qrelator{\HowClo{\arbsim}} |\Stackt\{\GenTermT'\}|$.
		
		\item $\GenTerm' = \caseN{\ValTerm}{\ComTerm}{x}{\ComTermT} : \ComType$, then $\GenTerm' \ComExt{\HowClo{\arbsim}} \GenTermT'$ is
		only concluded from \textbf{C4}, 
		hence $\GenTermT'$ is of the form $\caseN{\ValTerm'}{\ComTerm'}{x}{\ComTermT'}$ 
		for which we have $\ValTerm \HowClo{\arbsim} \ValTerm'$, $\ComTerm \HowClo{\arbsim} \ComTerm'$ and $\{x\} \vdash^c \ComTermT \HowClo{\arbsim} \ComTermT'$. 
		Since $\ValTerm \HowClo{\arbsim} \ValTerm'$ we have $\ValTerm = \ValTerm'$ by Lemma \ref{Nequ}. We do a case distinction on $\ValTerm$. 
		\begin{enumerate}
			\item If $\ValTerm = \zero = \ValTerm'$, then $\reduct{\GenTerm'}{\ComTerm}$ 
			so by Lemma \ref{lemma:frame_red} $|\GenTerm|_{n+1} = |\Stack\{\ComTerm\}|_n$ where since $\Stack$ matches $\Stackt$, 
			by induction hypothesis and $\ComTerm \HowClo{\arbsim} \ComTerm'$ we have 
			
			$|\Stack\{\ComTerm\}|_n \Qrelator{\HowClo{\arbsim}} |\Stackt\{\ComTerm'\}| = |\Stackt\{\caseN{\zero}{\ComTerm'}{x}{\ComTermT'}\}| = |\Stackt\{\GenTermT'\}|$. 
		
			\item If $\ValTerm = \succe{\ValTermT} = \ValTerm'$, 
			then $\reduct{\GenTerm'}{\ComTermT[\ValTermT/x]}$ 
			so by Lemma \ref{lemma:frame_red} $|\GenTerm|_{n+1} = |\Stack\{\ComTermT[\ValTermT/x]\}|_n$ 
			where by  $\{x\} \vdash^c \ComTermT \HowClo{\arbsim} \ComTermT'$ we have $\ComTermT[\ValTermT/x] \HowClo{\arbsim} \ComTermT'[\ValTermT/x]$ (Lemma \ref{Hprop1}). 
			By matching frames and induction hypothesis, 
			$|\Stack\{\ComTermT[\ValTermT/x]\}|_n \Qrelator{\HowClo{\arbsim}} |\Stackt\{\ComTermT'[\ValTermT/x]\}| = |\Stackt\{\GenTermT'\}|$.
		\end{enumerate}
	
		\item $\GenTerm' = \valseq{\ValTerm}{x}{\ComTerm} : \ComType$. 
		The only premises for $\GenTerm' \ComExt{\HowClo{\arbsim}} \GenTermT'$ are given by \textbf{C6}, 
		from which we know $\GenTermT' = \valseq{\ValTermT}{x}{\ComTermT}$, where $\ValTerm \HowClo{\arbsim} \ValTermT$ and $\{x\} \vdash \ComTerm \HowClo{\arbsim} \ComTermT$. 
		From this and Lemma \ref{Hprop1} we have $\ComTerm[\ValTerm/x] \HowClo{\arbsim} \ComTermT[\ValTermT/x]$ hence by matching frames and induction hypothesis we get $|\Stack\{\ComTerm[\ValTerm/x]\}|_n \Qrelator{\HowClo{\arbsim}} |\Stackt\{\ComTermT[\ValTermT/x]\}|$. 
		We have $\reduct{\GenTerm'}{\ComTerm[\ValTerm/x]}$ hence by Lemma \ref{lemma:frame_red} $|\GenTerm|_{n+1} = |\Stack\{\ComTerm[\ValTerm/x]\}|_n$, similarly $|\Stackt\{\GenTermT'\}| = |\Stackt\{\ComTermT[\ValTermT/x]\}|$ so $|\GenTerm|_{n+1}  \Qrelator{\HowClo{\arbsim}} |\Stackt\{\GenTermT'\}|$.
		
		\item $\GenTerm' = \fix{\ComTerm} : \ValTypeT \rightarrow \ComType$, we have $\fix{\ComTerm} \ComExt{\HowClo{\arbsim}} \GenTermT'$ which can only be concluded from \textbf{C19}, hence $\GenTermT' = \fix{\ComTerm'}$ where $\ComTerm \HowClo{\arbsim} \ComTerm'$. Look at the term $\{x\} \vdash^c \ComTermT = \appt{x}{\thunk{\fix{x}}}$, 
		by reflexivity we have $\{x\} \vdash^c \ComTermT \HowClo{\arbsim} \ComTermT$ and hence by Lemma \ref{Hprop1} we have $\ComTermT[\ComTerm/x] \HowClo{\arbsim} \ComTermT[\ComTerm'/x]$. 
		Using matching frames and induction hypothesis we can derive $|\Stack\{\ComTermT[\ComTerm/x]\}|_n \Qrelator{\HowClo{\arbsim}} |\Stackt\{\ComTermT[\ComTerm'/x]\}|$. Since $\reduct{\GenTerm}{\Stack\{\ComTermT[\ComTerm/x]\}}$ we use Lemma \ref{lemma:frame_red} to get
		$|\GenTerm|_{n+1} = |\Stack\{\ComTermT[\ComTerm/x]\}|_n$. Combining this with $|\Stackt\{\GenTermT'\}| = |\Stackt\{\ComTermT[\ComTerm'/x]\}|$ we get the desired conclusion.
	
	\enumeratext{Pattern match cases.}
	
		\item $\GenTerm' = \patmat{\ValTerm}{\{\dots,(i.x)\,.\,\ComTerm_i,\dots\}}$, then $\GenTerm' \ComExt{\HowClo{\arbsim}} \GenTermT'$ can only be from rule \textbf{C14}, hence $\GenTermT' = \patmat{\ValTermT}{\{\dots,(i.x)\,.\,\ComTermT_i,\dots\}}$ where $\ValTerm \HowClo{\arbsim} \ValTermT$ and for each $i$ we have $\{x\} \vdash^c \ComTerm_i \HowClo{\arbsim} \ComTermT_i$.
		As a value from a sum-type, $\ValTerm$ and $\ValTermT$ must be of the shape $\pair{j}{\ValTerm'}$ and $\pair{k}{\ValTermT'}$ respectively, where by Lemma \ref{lemma:How_simul_sum} we have $j=k$ and $\ValTerm' \HowClo{\arbsim} \ValTermT'$.
		Hence $\reduct{\GenTerm'}{\ComTerm_j[\ValTerm'/x]}$ and $\reduct{\GenTermT'}{\ComTermT_j[\ValTermT'/x]}$, where by Lemma \ref{Hprop1} $\ComTerm_j[\ValTerm'/x] \HowClo{\arbsim} \ComTermT_j[\ValTermT'/x]$.
		We conclude that by induction hypothesis that
		$|\Stack\{\GenTerm'\}|_{n+1} = |\Stack\{\ComTerm_j[\ValTerm'/x]\}|_n \Qrelator{\HowClo{\arbsim}} |\Stackt\{\ComTermT_j[\ValTermT'/x]\}| = |\Stackt\{\GenTermT'\}|$.
		
		\item $\GenTerm' = \patmat{\ValTerm}{\pair{x}{y}.\ComTerm}$, then $\GenTerm' \ComExt{\HowClo{\arbsim}} \GenTermT'$ can only be from rule \textbf{C16}, hence\\ $\GenTermT' = \patmat{\ValTermT}{\pair{x}{y}.\ComTermT}$ where $\ValTerm \HowClo{\arbsim} \ValTermT$ and $\{x,y\} \vdash^c \ComTerm \HowClo{\arbsim} \ComTermT$.
		As a value of pair-type, $\ValTerm$ and $\ValTermT$ must be of the shape $\pair{\ValTerm_0}{\ValTerm_1}$ and $\pair{\ValTermT_0}{\ValTermT_1}$ respectively, where by Lemma \ref{lemma:How_simul_pair} we have $\ValTerm_0 \HowClo{\arbsim} \ValTermT_0$ and $\ValTerm_1 \HowClo{\arbsim} \ValTermT_1$.
		Hence $\reduct{\GenTerm'}{\ComTerm[\ValTerm_0/x,\ValTerm_1/y]}$ and $\reduct{\GenTermT'}{\ComTermT[\ValTermT_0/x,\ValTermT_1/y]}$, and by Lemma \ref{Hprop1} $\ComTerm[\ValTerm_0/x,\ValTerm_1/y] \HowClo{\arbsim} \ComTermT[\ValTermT_0/x,\ValTermT_1/y]$.
		We conclude by Lemma \ref{lemma:frame_red} and the induction hypothesis that:
		
		$|\Stack\{\GenTerm'\}|_{n+1} = |\Stack\{\ComTerm[\ValTerm_0/x,\ValTerm_1/y]\}|_n \Qrelator{\HowClo{\arbsim}} |\Stackt\{\ComTermT[\ValTermT_0/x,\ValTermT_1/y]\}| = |\Stackt\{\GenTermT'\}|$.
			
	\enumeratext{Lastly the effect constructor cases.}
		\item $\GenTerm' = \EfOp(\ComTerm_0,\ComTerm_1,...) : \ComType$. Then $\GenTerm' \ComExt{\HowClo{\arbsim}} \GenTermT'$ can only be from \textbf{C20}, hence\\ $\GenTermT' = \EfOp(\ComTerm_0',\ComTerm_1',...)$ where for all $i$, $\ComTerm_i \HowClo{\arbsim} \ComTerm_i'$. 
		By matching frames and induction hypothesis, $|\Stack\{\ComTerm_i\}|_n \Qrelator{\HowClo{\arbsim}} |\Stackt\{\ComTerm_i'\}|$, 
		hence by Lemma \ref{lemma:frame_red}:
		
		$|\GenTerm|_{n+1} = |\Stack\{\GenTerm'\}|_{n+1} = |\EfOp(\Stack\{\ComTerm_0\},\Stack\{\ComTerm_1\},...)|_{n+1} = \EfOp \{i \mapsto |\Stack\{\ComTerm_i\}|_n\} \Qrelator{\HowClo{\arbsim}} \EfOp \{i \mapsto |\Stackt\{\ComTerm_i'\}|\} = |\Stackt\{\GenTermT'\}|$ by Cor. \ref{relat:preserve} and Lemma \ref{lemma:frame_red}. 
		
		\item $\GenTerm' = \EfOp(x \,.\, \ComTerm) : \ComType$ where $x : \NatType \vdash^c \ComTerm : \ComType$. Then $\GenTerm' \ComExt{\HowClo{\arbsim}} \GenTermT'$ can only be from \textbf{C21}, hence $\GenTermT' = \EfOp(x \,.\, \ComTermT)$ where $\ComTerm \HowClo{\arbsim} \ComTermT$. 
		Hence for all $\numeral{k} : \NatType$ we have $\ComTerm[\numeral{k}/x] \HowClo{\arbsim} \ComTermT[\numeral{k}/x]$ 
		and hence by matching frames and induction hypothesis $|\Stack\{\ComTerm[\numeral{k}/x]\}|_n \Qrelator{\HowClo{\arbsim}} |\Stackt\{\ComTermT[\numeral{k}/x]\}|$. Hence by Lemma \ref{lemma:frame_red} and Cor. \ref{relat:preserve} it holds that
		
		$|\GenTerm|_{n+1} = \EfOp(k \mapsto |\Stack\{\ComTerm[\numeral{k}/x]\}|_n) \Qrelator{\HowClo{\arbsim}} \EfOp(k \mapsto |\Stackt\{\ComTermT[\numeral{k}/x]\}|) = |\Stackt\{\GenTermT'\}|$.
		
		\item In the cases where $\GenTerm'$ is $\EfOp(\ValTerm,...)$, from compatibility rules \textbf{C22} we know $\GenTermT' = \EfOp(\ValTerm',...)$ where $\ValTerm \HowClo{\arbsim} \ValTerm'$ hence $\ValTerm = \ValTerm'$. The rest of the proof is similar to case (xi).
	\end{enumerate}
	That finishes the case distinction, so we know that for any shape of $\GenTerm'$ it holds that $|\Stack\{\GenTerm'\}|_{n+1} \Qrelator{\HowClo{\arbsim}} |\Stackt\{\GenTermT'\}|$. As was discussed before, this is sufficient in establishing that $|\GenTerm|_{n+1} \Qrelator{\HowClo{\arbsim}} |\GenTermT|$, and hence this finishes the induction step. So the proof by induction has been finished.
\end{proof}

We can conclude that $\ComTerm \HowClo{\arbsim} \ComTermT \Rightarrow |\ComTerm| \Qrelator{\HowClo{\arbsim}} |\ComTermT|$ for closed terms of type $\ForceType{\ValType}$. As such, we can conclude:

\begin{lemma}[Main Lemma]
	Given a Scott open and decompositional $\QObser$, then the Howe closure of a $\QObser$-simulation is a $\QObser$ simulation.
\end{lemma}

In particular, the Howe's closure of $\QObser$-similarity is a $\QObser$-simulation, and hence the $\QObser$-similarity itself.
Since the Howe's closure of a preorder is itself compatible, we can can conclude that the $\QObser$-similarity is compatible.
We can now derive Theorem \ref{theorem:bisim} as stated in Section \ref{section:Howe}, with the same method as in \cite{modal}.
The bisimilarity part of this result is established using what is known as the \emph{transitive closure trick} (see e.g. \cite{modal}).

\end{document}